\documentclass[letter,11pt]{article}
\usepackage{graphicx}
\usepackage{subcaption}
\usepackage{verbatim}
\usepackage{titling}
\usepackage{cancel}
\usepackage{enumitem}
\setlist{noitemsep,topsep=0pt,parsep=0pt}
\usepackage[margin=1cm]{caption}
\usepackage{mathtools}
\usepackage{multirow}
\usepackage[all]{xy}
\usepackage{tikz}
\usepackage{textcomp}
\usepackage{listings}
\usepackage{xcolor}

\usetikzlibrary{shapes.geometric, arrows}

\definecolor{codegreen}{rgb}{0,0.6,0}
\definecolor{codegray}{rgb}{0.5,0.5,0.5}
\definecolor{codepurple}{rgb}{0.58,0,0.82}
\definecolor{backcolour}{rgb}{0.95,0.95,0.92}

\lstdefinestyle{mystyle}{
    backgroundcolor=\color{backcolour},   
    commentstyle=\color{codegreen},
    keywordstyle=\color{magenta},
    numberstyle=\tiny\color{codegray},
    stringstyle=\color{codepurple},
    basicstyle=\ttfamily\footnotesize,
    breakatwhitespace=false,         
    breaklines=true,                 
    captionpos=b,                    
    keepspaces=true,                 
    numbers=left,                    
    numbersep=5pt,                  
    showspaces=false,                
    showstringspaces=false,
    showtabs=false,                  
    tabsize=2
}

\usetikzlibrary{arrows,backgrounds,calc,fit,decorations.pathreplacing,decorations.markings,shapes.geometric}
\tikzset{every fit/.append style=text badly centered}

\usepackage{tyson}

\usepackage[bookmarks=true,hypertexnames=false]{hyperref}
\hypersetup{colorlinks=true, citecolor=blue, linkcolor=red, urlcolor=blue}
\makeatletter

\newcommand{\Rmnum}[1]{\expandafter\@slowromancap\romannumeral #1@}
\makeatother

\newcommand{\Holant}{\operatorname{Holant}}

\newcommand{\holant}[2]{\ensuremath{\Holant\left(#1\mid #2\right)}}

\newcommand{\numP}{{\rm \#P}}

{\par\medskip}

\newcommand{\rlm}{\lambda/\mu} 
\newcommand{\smm}[4]{\left[ \begin{smallmatrix} #1 & #2 \\ #3 & #4 \end{smallmatrix}\right]}
\newcommand{\smmv}[2]{\left[ \begin{smallmatrix} #1 \\ #2 \end{smallmatrix}\right]}
\newcommand{\smmh}[2]{\left[ \begin{smallmatrix} #1 & #2 \end{smallmatrix}\right]}

\tikzstyle{internal} = [draw, fill, shape=circle]
\tikzstyle{external} = [shape=circle]
\tikzstyle{square}   = [draw, fill, rectangle]
\tikzstyle{triangle} = [draw, fill, regular polygon, regular polygon sides=3, inner sep=3pt]
\tikzstyle{pentagon} = [draw, fill, regular polygon, regular polygon sides=5, inner sep=2pt, minimum size=14pt]

\begin{document}

\title{{\bf Title}}

\vspace{0.3in}

	\title{\bf Bipartite 3-Regular Counting Problems with Mixed Signs}
		\vspace{0.3in}
		\author{Jin-Yi Cai\thanks{Department of Computer Sciences, University of Wisconsin-Madison. Supported by NSF CCF-1714275.
			} \\ {\tt jyc@cs.wisc.edu} \and Austen Z. Fan\thanks{Department of Computer Sciences, University of Wisconsin-Madison.}\\ {\tt afan@cs.wisc.edu}
\and  Yin Liu\thanks{Department of Computer Sciences, University of Wisconsin-Madison.}\\ {\tt yinl@cs.wisc.edu}
		}
	
\date{}
\maketitle

\begin{abstract}
We prove a complexity dichotomy for
a class of counting problems expressible as bipartite 3-regular Holant problems.
For every problem of the form \holant{f}{=_3}, where $f$ is any
integer-valued ternary symmetric constraint function on Boolean variables, we prove that it is 
either P-time computable or \#P-hard,
depending on an explicit criterion of $f$. The constraint function  can take both positive and negative values, allowing for cancellations. 
In addition, we 
discover a new phenomenon: there is a set $\mathcal{F}$ with the property that for every $f \in \mathcal{F}$ the problem \holant{f}{=_3} is planar P-time  computable but \numP-hard in general, yet 
 its planar tractability is by
a \emph{combination} of a holographic transformation by
$\left[\begin{smallmatrix} 1 & 1 \\ 1 & -1 \end{smallmatrix}\right]$ to FKT \emph{together} with an independent global argument. 
 \end{abstract}

\section{Introduction}
Holant problems encompass a broad class of
counting problems~\cite{Backens17,Backens18,BackensG20,CaiGW16,CaiL11, cai2009holant, CaiLX10,GuoHLX11,GuoLV13,KowalczykC10,valiant2006accidental,valiant2008holographic,Xia11}. For symmetric constraint functions this is also equivalent to
edge-coloring models~\cite{szegedy2007edge,szegedy2010edge}.
These problems extend counting constraint satisfaction problems. Freedman, Lov\'{a}sz and Schrijver proved that some prototypical Holant problems, such as counting perfect matchings, cannot be expressed as vertex-coloring models known as graph homomorphisms~\cite{Freedman-Lovasz-Schrijver-2007,HellN04}. 
The classification program of counting problems is to 
classify as broad a class of these problems as possible
into either \#P-hard or P-time computable.

While much progress has been made for the
classification of counting CSP~\cite{bulatov2006dichotomy,jacm/CaiC17,cai2016nonnegative,dyer2013effective}, and some progress for
Holant problems~\cite{cai2017complexity}, classifying Holant problems on regular bipartite
graphs is particularly challenging.
In a very recent paper~\cite{fan2020dichotomy} we initiated the
study of Holant problems on 
the simplest setting of
3-regular bipartite graphs with \emph{nonnegative}
constraint functions. Admittedly, this is a severe restriction,
because nonnegativity of the constraint functions
rules out cancellation, which is a source of non-trivial
P-time algorithms. Cancellation is in a sense  the 
\emph{raison d'\^{e}tre} for the Holant framework following
Valiant's holographic algorithms~\cite{valiant2006accidental,valiant2008holographic, valiant2018some}. The (potential) existence
of  P-time algorithms by cancellation
is exciting, but at the same time creates obstacles
if we want to classify every problem in the family
into either P-time computable or \#P-hard.
At the same time, restricting to nonnegative
constraint functions makes
the classification theorem easier to prove.
In this paper, we remove this nonnegativity restriction.

More formally, a Holant problem is defined on a graph where edges are
variables and vertices are constraint functions.
The aim of a Holant problem 
is to compute its partition function, which is
a sum over all $\{0, 1\}$-edge assignments of the product over
all vertices of the  constraint function evaluations.
E.g., if every vertex has the 
\textsc{Exact-One}
function (which evaluates to 1 if
exactly one incident edge is 1, and evaluates to  0 otherwise),
then the partition function gives
the number of perfect matchings.
In this paper we consider Holant problems on 
3-regular bipartite graphs $G=(U,V,E)$, where
 the Holant problem \holant{f}{=_3} computes the following 
 partition function\footnote{If we replace $f$ by a set 
 ${\cal F}$ of constraint functions, each $u \in U$
 is  assigned some $f_u \in {\cal F}$,
 and replace $(=_3)$ by ${\cal EQ}$,  the set
 of \texttt{Equality} of all arities, then 
\holant{{\cal F}}{{\cal EQ}} can be taken as the definition of counting
CSP.}
$$\Holant (G)=\sum_{\sigma: E \rightarrow\{0,1\}} \prod_{u \in U} f\left(\left.\sigma\right|_{E(u)}\right) \prod_{v \in V}(=3)\left(\left.\sigma\right|_{E(v)}\right),$$
where $f = [f_0, f_1, f_2, f_3]$ at each $u \in U$ 
is an integer-valued constraint function
that evaluates to $f_i$ if $\sigma$ assigns exactly $i$ 
among 3 incident edges  $E(u)$ to 1, and $(=_3) = [1,0,0,1]$
 is the \textsc{Equality} function on 3 variables (which is 1 iff
 all three are equal).
 E.g., if we take the 
\textsc{Exact-One}
function  $f=[0,1,0,0]$ then \holant{f}{=_3} 
 counts the number of exact-3-covers;   if $f$ is
 the  
\textsc{Or}
function $[0, 1 ,1 ,1]$  then \holant{f}{=_3} 
 counts the number of all set covers.
 
The main theorem in this paper is a complexity dichotomy
(Theorem~\ref{thm:main}): for any rational-valued 
function $f$ of arity 3, the problem
 \holant{f}{=_3} is either \numP-hard  or P-time
 computable, depending on an explicit
 criterion on $f$.
The main advance is to allow $f$ to take both positive and
negative values, thus cancellations in the sum
$\sum_{\sigma: E \rightarrow\{0,1\}}$ can occur.


A major component of the classification program is
to account for some algorithms, called holographic
algorithms, that were initially discovered by Valiant~\cite{valiant2006accidental}. These algorithms introduce quantum-like cancellations as
the main
tool.
In the past 10 to 15 years we have gained a great
deal of understanding of these mysteriously looking algorithms. In particular, it was proved
in~\cite{CaiLX10} that for all counting CSP
with arbitrary constraint functions on Boolean variables,
there is a precise 3-way division of problem types: 
(1) 
P-time computable in general, (2) P-time
computable on planar structures but \#P-hard in general, and (3)
\#P-hard even on planar structures.
Moreover, every problem in type (2) is computable in P-time on planar structures 
by Valiant's holographic reduction to Fisher-Kasteleyn-Temperley algorithm (FKT) for planar
perfect matchings. In~\cite{CaiFGW15}
for (non-bipartite) Holant problems with symmetric
constraint functions, the 3-way division above persists,
but problems  in (2) includes one more subtype unrelated
to Valiant's holographic reduction.
In this paper, we have a
surprising discovery. We found a 
 new set of functions $\mathcal{F}$ which fits into type (2)
 problems above, but the planar P-time tractability is 
 \emph{neither} by  Valiant's holographic reduction alone,
  \emph{nor} entirely independent from it. Rather it is by a 
  combination of a holographic reduction together with
  a global argument.
  An example of this set of problem is as follows:
We say  $(X, {\cal S})$  is a 3-regular $k$-uniform set system,
if ${\cal S}$ consists of a family of sets $S \subset X$ each of
size $|S| =k$, and every $x \in X$ is in exactly 3 sets.
If $k=2$ this is just a 3-regular graph.  We consider
 3-regular $3$-uniform set systems.
 We say ${\cal S'}$ is a \emph{leafless partial cover} 
 if every $x \in \bigunion_{S \in {\cal S'}} S$ belongs to more than
 one set $S  \in {\cal S'}$. We say $x$ is \emph{lightly covered}
 if $|\{S \in {\cal S'} : x \in S\}|$ is 2,
 and \emph{heavily covered} if this number is 3.
  
  \vspace{.1in}
\noindent$\mathbf{Problem:}$  \texttt{ Weighted-Leafless-Partial-Cover}.

\noindent$\mathbf{Input:}$ A 3-regular $3$-uniform set system
 $(X, {\cal S})$.

\noindent$\mathbf{Output:}$ $\sum_{\cal S'} (-1)^l 2^h$,
where the sum is over all leafless partial covers  ${\cal S'}$,
and $l$ (resp. $h$) is the number of $x \in X$ that are
{lightly covered} (resp. {heavily covered}).
  \vspace{.1in}

One can show that this problem is just
\holant{f}{=_3}, where $f = [1, 0, -1, 2]$.
This problem is a special case of a set of problems of
 the form $f = [3a+b, a-b, -a+b, 3a-b]$.
We show that all these problems belong to type (2) above, although they are not directly solvable by
a holographic algorithm since they are provably
not matchgates-transformable.


In this paper, we 
 use  Mathematica\texttrademark{} 
to perform symbolic computation. In particular, the procedure \texttt{CylindricalDecomposition}  in Mathematica\texttrademark{}  is an implementation (of a version) of  Tarski's theorem
on the decidability of the theory of
real-closed fields. 
Some of our proof steps involve heavy symbolic computation. 
This stems from the bipartite structure. In order to   
preserve this structure, one has to connect each vertex from LHS to only vertices from RHS when constructing subgraph fragments called gadgets. In 3-regular bipartite graphs, it is easy to show that any  gadget construction produces a constraint function that  has the following restriction: the difference of the  arities  between the two sides is 0 mod 3. This severely limits the
possible  constructions within a moderate size, and a reasonable sized construction tends to produce gigantic polynomials.  To ``solve'' some of these
polynomials seems beyond direct manipulation by hand. 

We believe our dichotomy 
(Theorem~\ref{thm:main})
is valid even for
(algebraic) real or complex-valued constraint functions.
However, in this paper we can only prove it for rational-valued
constraint functions.
There are two difficulties of extending our proof
beyond $\mathbb{Q}$. The first is that 
we use the idea of interpolating degenerate straddled functions,
for which we need to ensure that the ratio of the eigenvalues
of the interpolating gadget matrix
is not a root of unity. With rational-valued constraint functions, the only  roots of unity
that can occur are in a degree 2 extension field.
For general constraint functions, they can be
arbitrary roots of unity. Another difficulty is
that some Mathematica\texttrademark{} steps
showing the nonexistence of some exceptional cases are only valid
for $\mathbb{Q}$.
We list the essential Mathematica\texttrademark{} procedures used in this proof in an appendix.

%

\section{Preliminaries}
We now introduce the concept of \emph{gadget}.
A gadget, such as those illustrated in Figure~\ref{f1} to Figure~\ref{2g41a},
is  a bipartite graph $G = (U, V, E_{\rm in}, E_{\rm out})$ with
internal edges $E_{\rm in}$ and dangling edges $E_{\rm out}$.
There can be $m$ dangling edges internally incident
to vertices from $U$ and $n$ dangling edges internally incident
to vertices from $V$. These $m+n$ dangling edges 
correspond to Boolean variables $x_1, \ldots, x_m, y_1, \ldots,  y_n$
and the gadget defines
a \emph{signature}
\[f(x_1, \ldots, x_m, y_1, \ldots,  y_n)
=
\sum_{\sigma: E_{\rm in} \rightarrow \{ 0,1\}} \prod_{u \in U} f\left(\widehat{\sigma} |_{E(u)}\right) \prod_{v \in V} \left( =_{3} \right) \left(\widehat{\sigma}  |_{E(v)}\right),
\]
where $\widehat{\sigma} $ denotes the extension
of $\sigma$ by the assignment on the  dangling edges.

As indicated before, in the setting of 3-regular bipartite graph, we have limited number of symmetric gadgets with reasonable sizes. To preserve the bipartite structure, we must be careful
in any gadget construction how each  external wire is
to be connected, i.e., as an input variable
whether it is on the LHS (those of $f$ which can be
used to  connect to $(=_3)$ on the RHS), or
it is on the RHS (those of $(=_3)$ which can be
used to  connect to $f$ on the LHS).

In each figure of gadgets presented later, we use a blue square to represent a signature from LHS, which under most of the cases will be $[1,a,b,c]$,  a green circle to represent the ternary equality $[1, 0, 0, 1]$, and a black triangle to represent a unary signature whose values depend on the context.


A \emph{signature grid} $\Omega=(G, \pi)$ over a signature set $\mathcal{F}$ consists of a graph $G=(V, E)$ and a mapping $\pi$ that assigns to each vertex $v \in V$ an $f_{v} \in \mathcal{F}$ and a linear order of the incident edges at $v$. For signature sets $\mathcal{F}$ and $\mathcal{G}$, a bipartite signature grid over $(\mathcal{F} \mid \mathcal{G})$ is a signature grid $\Omega=$ $(H, \pi)$ over $\mathcal{F} \cup \mathcal{G}$, where $H=(V, E)$ is a bipartite graph with bipartition $V=$ $\left(V_{1}, V_{2}\right)$ such that $\pi\left(V_{1}\right) \subseteq \mathcal{F}$ and $\pi\left(V_{2}\right) \subseteq \mathcal{G}$. In this paper, we consider the bipartite Holant problem where $\mathcal{F} = \{f\}$ consists of a single rational-valued ternary symmetric Boolean function and $\mathcal{G}=\{[1,0,0,1]\}$ consists of $\textsc{Equality}_3$.

A \emph{symmetric} signature is a function that is invariant under
any permutation of its variables.
The value of such a signature  depends only on the Hamming weight of its input. We denote a ternary symmetric signature $f$
by the notation $f = [f_0,f_1,f_2,f_3]$, where 
$f_i$ is the value on inputs of 
Hamming weight $i$.
%
The \textsc{Equality} of arities 3 is $\left(=_{3} \right) = [1,0,0,1]$. A  symmetric signature $f$ is called (1) \emph{degenerate} if it is the tensor power of
a unary signature; (2) \emph{Generalized Equality}, or Gen-Eq, if it is zero unless all inputs are equal.
Affine signatures were discovered in the dichotomy
for counting constraint satisfaction problems (\#CSP)~\cite{cai2017complexity}. 
A (real valued) ternary
 symmetric signature is
\emph{affine} if it has the form $[1,0,0,\pm 1], [1,0,1,0], [1,0,-1,0], [1,1,-1,-1]$ or $[1,-1,-1,1]$, or by reversing the order of the entries,
up to a constant factor. If $f$ is degenerate, Gen-Eq, or affine, then the problem $\#\operatorname{CSP}(f)$ and thus $\holant{f}{=_3}$ is in P (for a more detailed exposition of
this theory,  see \cite{cai2017complexity}). Our dichotomy asserts that,
for all signatures
$f$ with $f_i \in \mathbb{Q}$,
these three classes are the only tractable cases of the problem $\holant{f}{=_3}$; all other signatures  lead to $\numP$-hardness. 

By a slight abuse of notation, we say $[1,a,b,c]$
is \numP-hard or in P depending on weather the 
 problem $\holant{[1,a,b,c]}{\left( =_3\right)}$ is 
\numP-hard or in P, respectively. We shall invoke the following theorem~\cite{KowalczykC10} when proving our results:
\begin{theorem} \label{2-3}
Suppose $a,b \in \mathbb{C}$, and let $X = ab$, $Z = \left( \frac{a^3+b^3}{2} \right)^2$. Then $\holant{[a,1,b]}{\left(=_3\right)}$ is \numP-hard except in the following cases, for which the problem is in $\operatorname{P}$.

\begin{enumerate}
\item $X=1$;
\item $X=Z=0$;
\item $X=-1$ and $Z=0$;
\item $X=-1$ and $Z=-1$.
\end{enumerate}
\end{theorem}




\begin{figure}[] 
\centering
  \begin{tikzpicture}  
        \draw [very thick] (-1,0.5) -- (0.5,0.5);
        
        \draw[very thick] (2,0.5) ellipse (1.5cm and 0.5cm);
        \draw[very thick] (3.5,0.5) -- (5,0.5);
        \filldraw[fill= blue] (0.4,0.4) rectangle (0.6,0.6);
        \filldraw[fill=green] (3.5,0.5) circle (0.1cm);
    \end{tikzpicture}
    \caption{$G_1$}
  \label{f1}
\end{figure}

An important observation is that
in the context of  $\holant{f}{\left( =_3\right)}$,
every gadget construction produces a signature
with $m \equiv n \bmod 3$, where $m$ and $n$ are the
numbers of input variables (arities) from the
LHS and RHS respectively. Thus, any construction that produces a signature purely on either the LHS or the RHS 
will have arity a multiple of 3. In order that our
constructions are more manageable in size, 
we will make heavy use of 
 \emph{straddled gadgets} with
$m=n=1$ that do not belong to either side and yet
can be easily iterated. The signatures of the iterated gadgets
are
 represented by matrix powers. 

Consider the binary straddled gadget $G_1$ in Figure \ref{f1}. Its signature is $G_1=\left[\begin{smallmatrix} 1 & b \\ a & c \end{smallmatrix}\right]$, where  $G_1(i,j)$ (at row $i$ column $j$) is the value of this gadget when the left dangling edge (from the ``square") and the right dangling edge (from the ``circle" $(=_3)$) are assigned $i$ and $j$ respectively, for $i,j \in\{0,1\}$.
Iterating $G_1$ sequentially $k$ times is represented by the
matrix power $G_1^k$. It turns out that it is very useful
 either to produce directly  or to obtain by interpolation
a rank \emph{deficient} straddled signature, which would in most cases allow us to obtain unary signatures on either side.
With unary signatures we can connect to a ternary signature
to produce binary signatures on one side and then apply
Theorem~\ref{2-3}. 
The proof idea of Lemma~\ref{3.1.1} is the same as
in~\cite{fan2020dichotomy} for nonnegative signatures.


\begin{lemma} \label{3.1.1}
Given the binary straddled signature $G_1=\left[\begin{smallmatrix} 1 & b \\ a & c \end{smallmatrix}\right]$, we can interpolate the degenerate binary straddled signature $\left[\begin{smallmatrix}y & xy \\ 1 & x\end{smallmatrix}\right]$,
provided that $c\neq ab$, $a\neq 0$, $\Delta =\sqrt{(1-c)^2 + 4ab} \neq 0$ and $\frac{\lambda}{\mu}$ is not a root of unity, where
$\lambda=\frac{-\Delta+(1+c)}{2}$, $\mu=\frac{\Delta+(1+c)}{2}$ are the
 two eigenvalues,  
 and $x=\frac{\Delta-(1-c)}{2 a}$ and $y=\frac{\Delta+(1-c)}{2 a}$.
\end{lemma}

\begin{proof}
We have $x+y = \Delta/a \not =0$ and so
$\left[\begin{smallmatrix} -x & y \\ 1 & 1 \end{smallmatrix}\right]^{-1}$ exists, and 
the matrix $G_1$ has the Jordan Normal Form
$$G_1 = \left(\begin{array}{ll}1 & b \\ a & c\end{array}\right)=\left(\begin{array}{cc}-x & y \\ 1 & 1\end{array}\right)\left(\begin{array}{ll}\lambda & 0 \\ 0 & \mu\end{array}\right)\left(\begin{array}{cc}-x & y \\ 1 & 1\end{array}\right)^{-1}.$$ Here the matrix $G_1$ is non-degenerate since $c\ne ab$, and so  $\lambda$ and $\mu$ are nonzero.
Consider
$$
D = \frac{1}{x+y} \left( \begin{array}{cc} y & xy \\ 1 & x \end{array}\right) = \left( \begin{array}{cc} -x &y \\  1&1 \end{array}\right)  \left( \begin{array}{cc} 0 &0 \\  0&1\end{array}\right) \left( \begin{array}{cc} -x &y \\  1&1 \end{array}\right)^{-1}.
$$
Given any signature grid $\Omega$ where the binary degenerate straddled signature $D$ appears $n$ times, we form gadgets $G_1^s$ where $0 \le s \le n$ by iterating the $G_1$ gadget $s$ times and replacing each occurrence of $D$ with $G_1^s$.
(Here for $s=0$ we simply replace each occurrence of $D$ by an edge.) Denote the resulting
signature grid as $\Omega_{s}$. We stratify the assignments 
in the Holant sum for $\Omega$  according to assignments to
$\left[ \begin{smallmatrix} 0 &0 \\  0&1\end{smallmatrix}\right]$

\begin{enumerate}
\item[-] $(0,0)$ $i$ times;
\item[-] $(1,1)$ $j$ times;
\end{enumerate}
with $i+j=n$; all  other assignments will contribute 0 in the $\operatorname{Holant}$ sum for $\Omega$. The same statement is true for 
each $\Omega_s$ with the matrix
$\left[ \begin{smallmatrix} \lambda^s & 0 \\ 0 &\mu^s \end{smallmatrix}\right]$. 
Let $c_{i,j}$ be the sum, in $\Omega$, over all such assignments of the products of evaluations of all other signatures other than that represented by the matrix $\left[ \begin{smallmatrix} 0 &0 \\  0&1\end{smallmatrix}\right]$, including the contributions from $\left[ \begin{smallmatrix} -x &y \\  1&1 \end{smallmatrix}\right]$ and its inverse. The same quantities $c_{ij}$ appear for each $\Omega_s$, independent of $s$, 
with the substitution of the matrix $\left[ \begin{smallmatrix} \lambda^s & 0 \\ 0 &\mu^s \end{smallmatrix}\right]$. Then, for $0 \le s \le n$,  we have
\begin{equation}\label{Vandermonde}
\operatorname{Holant}_{\Omega_{s}} = \sum_{i+j=n} \left( \lambda^i \mu^j \right)^s \cdot c_{i,j}
\end{equation}
and
$
\operatorname{Holant}_{\Omega} = c_{0,n}.
$

Since $\lambda/\mu$ is not a root of unity,
the quantities $\lambda^i \mu^{n-i}$ are pairwise distinct, thus
(\ref{Vandermonde}) 
is a full ranked Vandermonde system. Thus we can compute
$
\operatorname{Holant}_{\Omega}
$
from 
$
\operatorname{Holant}_{\Omega_s}
$
by solving the linear system in polynomial time. Thus we can
interpolate $D$ in polynomial time.
\end{proof}

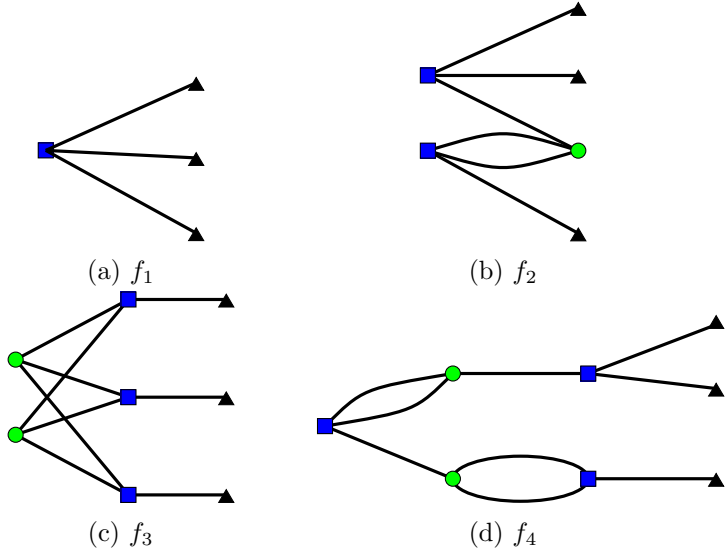
\begin{figure}[]
    \centering
    \begin{subfigure}[b]{0.3\textwidth}
    \centering
    \begin{tikzpicture}  
        \filldraw[fill=blue] (0.2,2.2) rectangle (0.4,2.4);
        \filldraw[fill=black] (2.3,3.27)--(2.2,3.1)--(2.4,3.1)--cycle ;
        \filldraw[fill=black] (2.3,2.27)--(2.2,2.1)--(2.4,2.1)--cycle;
        \filldraw[fill=black] (2.3,1.27)--(2.2,1.1)--(2.4,1.1)--cycle;
        \draw[very thick] (0.3,2.3)--(2.3,3.2);
        \draw[very thick] (0.3,2.3)--(2.3,2.2);
        \draw[very thick] (0.3,2.3)--(2.3,1.2);
    \end{tikzpicture}
    \caption{$f_1$}
    \label{e1}
    \end{subfigure}
    \begin{subfigure}[b]{0.3\textwidth}
    \centering
    \begin{tikzpicture}   
        \filldraw[fill=black] (2.3,5.27)--(2.2,5.1)--(2.4,5.1)--cycle;
        \filldraw[fill=black] (2.3,4.37)--(2.2,4.2)--(2.4,4.2)--cycle;
        \filldraw[fill=black] (2.3,2.27)--(2.2,2.1)--(2.4,2.1)--cycle;

        \draw[very thick] (0.3,4.3) -- (2.3,5.2);
        \draw[very thick] (0.3,4.3)--(2.3,4.3);
        \draw[very thick] (0.3,4.3)--(2.3,3.3);
        \draw[very thick] (0.3,3.3) ..controls (1.3,3.6)..  (2.3, 3.3);
        \draw[very thick] (0.3,3.3) .. controls (1.3, 3.0).. (2.3, 3.3);
        \draw[very thick] (0.3,3.3) -- (2.3, 2.2);
        \filldraw[fill=blue] (0.2,4.2) rectangle (0.4,4.4) ;  
        \filldraw[fill=blue] (0.2,3.2) rectangle (0.4,3.4) ;  
        \filldraw[fill=green] (2.3,3.3) circle (0.1cm) ;
    \end{tikzpicture}
    \caption{$f_2$}
    \label{e2}
    \end{subfigure}
    \\
    \begin{subfigure}[b]{0.3\textwidth}
    \centering
    \begin{tikzpicture}  
        \draw[very thick] (1.5,3.5)--(3,4.3);
        \draw[very thick] (1.5,3.5)--(3,3);
        \draw[very thick] (1.5,3.5)--(3,1.7);
        \draw[very thick] (1.5,2.5)--(3,4.3);
        \draw[very thick] (1.5,2.5)--(3,3);
        \draw[very thick] (1.5,2.5)--(3,1.7);
        \draw[very thick] (3,4.3)--(4.3, 4.3);
        \draw[very thick] (3,3)--(4.3, 3);
        \draw[very thick] (3,1.7)--(4.3,1.7);
        \filldraw[fill=green] (1.5,3.5) circle (0.1cm);  
        \filldraw[fill=green] (1.5,2.5) circle (0.1cm);  
        \filldraw[fill=blue] (2.9,4.2) rectangle (3.1,4.4); 
        \filldraw[fill=blue] (2.9,2.9) rectangle (3.1,3.1); 
        \filldraw[fill=blue] (2.9,1.6) rectangle (3.1,1.8); 
        \filldraw[fill=black] (4.3,4.37)--(4.2,4.2)--(4.4,4.2)--cycle; 
        \filldraw[fill=black] (4.3,3.07)--(4.2,2.9)--(4.4,2.9)--cycle; 
        \filldraw[fill=black] (4.3,1.77)--(4.2,1.6)--(4.4,1.6)--cycle; 
    \end{tikzpicture}
    \caption{$f_3$}
    \label{e3}
    \end{subfigure}
    \begin{subfigure}[b]{0.3\textwidth}
    \centering
    \begin{tikzpicture}  
        \draw[very thick] (0.3,4.3)..controls(0.8,4.8)..(2,5);
        \draw[very thick] (0.3,4.3)..controls(1.5, 4.5 )..(2,5);
        \draw[very thick] (0.3,4.3)--(2,3.6);
        \draw[very thick] (3.8,5)--(2,5);
        
        \draw[very thick] (3.8,5)--(5.5,5.65);
        \draw[very thick] (3.8,5)--(5.5,4.8);
        \draw[very thick] (3.8,3.6)--(5.5,3.6);
        
        \draw[very thick]  (2.9, 3.6) ellipse(0.9 and 0.3); 
        \filldraw[fill=blue] (0.2,4.2) rectangle(0.4,4.4); 
        \filldraw[fill=green] (2,5) circle (0.1cm);     
        \filldraw[fill=green] (2,3.6) circle(0.1cm);    
        \filldraw[fill=blue] (3.7,4.9) rectangle (3.9,5.1);   
        \filldraw[fill=blue] (3.7,3.5) rectangle (3.9, 3.7);  
        \filldraw[fill=black] (5.5,5.77)--(5.4,5.6)--(5.6,5.6)--cycle; 
        \filldraw[fill=black] (5.5,4.87)--(5.4,4.7)--(5.6,4.7)--cycle; 
        \filldraw[fill=black] (5.5,3.67)--(5.4,3.5)--(5.6,3.5)--cycle; 
    \end{tikzpicture}
    \caption{$f_4$}
    \label{e4}
    \end{subfigure}
    \caption{Four gadgets where each triangle represents the unary gadget $[1,x]$}
    \label{4g41x}
\end{figure}

The next lemma allows us to get unary signatures.

\begin{lemma} \label{3.1.2}
For $\Holant( \, [1,a,b,c] \, | =_3)$,\ 
$a,b,c \in \mathbb{Q}$, $a\neq 0$,  with the availability of binary degenerate straddled signature $\left[\begin{smallmatrix}y & xy \\ 1 & x\end{smallmatrix}\right]$ (here $x,y \in \mathbb{C}$ can be arbitrary),
in polynomial time 
\begin{enumerate}
    \item 
    we can interpolate $[y,1]$ on the LHS, 
    \[\Holant(\{[1,a,b,c], [y,1]\}\, | =_3) \le_T \Holant([1,a,b,c] | =_3); \]
    \item we can interpolate $[1,x]$ on the RHS, 
    \[\Holant([1,a,b,c] \, | \{(=_3), [1, x]\}) \le_T \Holant([1,a,b,c] | =_3) ,\] except for two cases: $[1,a,a,1]$, $[1,a,-1-2a, 2+3a]$.
\end{enumerate}
\end{lemma}

\begin{proof}
For the problem $\Holant(\{[1,a,b,c], [y,1]\}\ | =_3)$, the number of occurrences of $[y,1]$ on LHS is 0 mod 3, say $3n$, since the other signatures are both of arity 3. Now, for each occurrence of $[y,1]$, we replace it with the binary straddled signature $\left[\begin{smallmatrix}y & xy \\ 1 & x\end{smallmatrix}\right]$,  leaving $3n$ dangling edges on RHS yet to be connected to LHS, each of which represents a unary signature $[1,x]$. We build a gadget to connect every triple of
such dangling edges. We claim that at least one of the connection gadgets in Figures \ref{e1}, \ref{e2}, \ref{e3} and \ref{e4} creates a nonzero global factor. The factors of these four gadgets are $f_1 = cx^3+3bx^2+3ax+1$, $f_2=(ab+c)x^3+ (3bc+2a^2+b)x^2 +(2b^2+ac + 3a)x+ab+1$, $f_3=(a^3+b^3+c^3)x^3+3(a^2+2ab^2+bc^2)x^2 + 3(a+2a^2b+b^2c)x + 1 + 2a^3+b^3$ and $f_4=(ab+2abc+c^3)x^3+(2a^2+b+2a^2c+3ab^2+bc+3b^2c)x^2+(3a+3a^2b+ac+2b^2+2b^2c+ac^2)x+1+2ab+abc$ 
respectively. By setting the four formulae to be 0 simultaneously, with $a\neq 0$, $a,b,c\in\mathbb{Q}$ and $x\in\mathbb{C}$, we found that there is no solution\footnotemark. 
Thus, we can always ``absorb'' the left-over $[1,x]$'s at
the cost of some easily computable nonzero global factor. 
\footnotetext{We use Mathematica, where a complex $x$ (or $y$) is written as $u+vi$, and the real and imaginary parts of $f_i$ are both set to 0.
The empty intersection of $f_1=f_2=f_3=f_4=0$ is proved by cylindrical decomposition, an algorithm for Tarski's theorem on real-closed fields.}

For the other claim on $[1,x]$ on RHS, i.e., $$\Holant([1,a,b,c]\ |\ \{(=_3), [1,x]\}) \le_T \Holant([1,a,b,c]\ | =_3)$$ we use a similar strategy to ``absorb'' the left-over copies of $[y,1]$ on the LHS by connecting them to $(=_3)$ in the gadgets in the Figures \ref{h1}, \ref{h2} or \ref{h3}.
These gadgets produce factors $g_1 =y^3 + 1$, $g_2 = y^3 + by^2+ay+c$ and $g_3 = y^3 + 3a^2y^2 + 3b^2y + c^2$ respectively. It can be directly checked that, for complex $y$, all these factors are 0 iff $y=-1$, and the signature has the form $[1,a,a,1]$ or $[1,a,-2a-1, 3a+2]$.
\end{proof}

\begin{figure}[]
    \centering
    \begin{subfigure}[b]{0.3\textwidth}
    \centering
    \begin{tikzpicture}  
        \draw[very thick] (0.3,5.05) -- (2.2,4.05);
        \draw[very thick] (0.3,4.05) -- (2.2,4.05);
        \draw[very thick] (0.3,3.05) -- (2.2,4.05);
        \filldraw[fill=black] (0.3,5.17)--(0.2,5)--(0.4,5)--cycle;
        \filldraw[fill=black] (0.3,4.17)--(0.2,4)--(0.4,4)--cycle;
        \filldraw[fill=black] (0.3,3.17)--(0.2,3)--(0.4,3)--cycle;
        \filldraw[fill=green] (2.2, 4.05)circle(0.1);
    \end{tikzpicture}
    \caption{$g_1$}
    \label{h1}
    \end{subfigure}
    \hfill
    \begin{subfigure}[b]{0.3\textwidth}
    \centering
    \begin{tikzpicture}  
        \draw[very thick] (0.3,5.05)--(1.8,4.8);
        \draw[very thick] (0.3,4.05)--(1.8,4.8);
        \draw[very thick] (0.3,3.05)--(1.8,3.5);
        \draw[very thick] (1.8,4.8) -- (3,4.5);
        \draw[very thick] (1.8,3.5).. controls (2.3, 4.2) ..(3,4.4);
        \draw[very thick] (1.8,3.5) .. controls (2.7, 3.7)..(3,4.4);
        \filldraw[fill=black] (0.3,5.17)--(0.2,5)--(0.4,5)--cycle;
        \filldraw[fill=black] (0.3,4.17)--(0.2,4)--(0.4,4)--cycle;
        \filldraw[fill=black] (0.3,3.17)--(0.2,3)--(0.4,3)--cycle;
        \filldraw[fill=green] (1.8,4.8) circle (0.1);
        \filldraw[fill=green] (1.8,3.5) circle (0.1);
        \filldraw[fill=blue] (3,4.3)rectangle(3.2,4.5);
    \end{tikzpicture}
    \caption{$g_2$}
    \label{h2}
    \end{subfigure}
    \hfill
    \begin{subfigure}[b]{0.3\textwidth}
    \centering
    \begin{tikzpicture}   
        \draw[very thick] (0.3,5.2)--(1.8,5.2);
        \draw[very thick] (0.3,4.2)--(1.8,4.2);
        \draw[very thick] (0.3,3.2)--(1.8,3.2);
        
        \draw[very thick] (1.8,5.2)--(3.3,4.6);
        \draw[very thick] (1.8,4.2)--(3.3,4.6);
        \draw[very thick] (1.8,3.2)--(3.3,4.6);
        \draw[very thick] (1.8,5.2)--(3.3,3.6);
        \draw[very thick] (1.8,4.2)--(3.3,3.6);
        \draw[very thick] (1.8,3.2)--(3.3,3.6);
        
        \filldraw[fill=blue] (3.2,4.5)rectangle(3.4,4.7);
        \filldraw[fill=blue] (3.2,3.5)rectangle(3.4,3.7);
        \filldraw[fill=black] (0.3,5.27)--(0.2,5.1)--(0.4,5.1)--cycle;
        \filldraw[fill=black] (0.3,4.27)--(0.2,4.1)--(0.4,4.1)--cycle;
        \filldraw[fill=black] (0.3,3.27)--(0.2,3.1)--(0.4,3.1)--cycle;
        \filldraw[fill=green] (1.8,5.2)circle(0.1);
        \filldraw[fill=green] (1.8,4.2)circle(0.1);
        \filldraw[fill=green] (1.8,3.2)circle(0.1);
    \end{tikzpicture}
    \caption{$g_3$}
    \label{h3}
    \end{subfigure}
    
    \caption{Three gadgets where each triangle represents the unary gadget $[y,1]$}
    \label{3g4y1}
\end{figure}
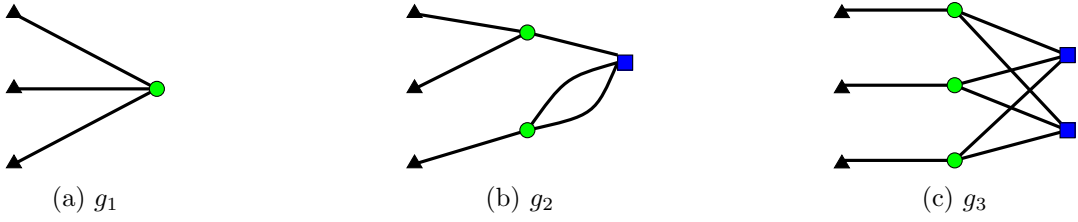
\noindent

A main thrust in our proof is we want to be assured
that such  degenerate binary straddled signature
can be obtained, and the corresponding unary signatures in 
Lemma~\ref{3.1.2} can be produced.
We now first consider the two exceptional cases $[1,a,a,1]$ and $[1,a,-2a-1,3a+2]$ where this is not possible.

\begin{lemma} \label{easy}
The problem $[1,a,a,1]$ is \numP-hard unless $a \in \{0, \pm 1\}$ in which case it is in {\rm P}.
\end{lemma}
\begin{proof}
If $a=0$ or $a=\pm 1$, then it is either Gen-Eq or degenerate or affine, and thus the problem $\Holant( \, [1,a,a, 1] \, | =_3)$ is in P. Now assume 
$a\neq 0$ and $a\neq \pm 1$.

Using the gadget $G_1$, we have $\Delta = |2a|$ and $x=y= \Delta/2a = \pm 1$ depending on the sign of $a$.
So we get the signature $[y,1] = [\pm 1, 1]$ on LHS by Lemmas~\ref{3.1.1} and \ref{3.1.2}. Connecting two copies of $[y,1]$ to $[1,0,0,1]$ on RHS, we get $[1,1]$ on RHS regardless of the sign. Connecting $[1,1]$ to $[1,a,a,1]$ on LHS, we get $[1+a,2a,1+a]$ on LHS. The problem
$\Holant( \, [1+a,2a,1+a] \, | =_3)$
 is \numP-hard  by Theorem~\ref{2-3} unless $a= 0, \pm 1$ or $-\frac{1}{3}$, thus we only need to consider the signature $[3, -1, -1, 3]$. 

If $a = -\frac{1}{3}$, we apply holographic transformation with the Hadamard matrix $H = \left[\begin{smallmatrix}1 & 1 \\ 1 & -1\end{smallmatrix}\right]$. Note that $[3,-1,-1,3] = 4((1,0)^{\otimes3} + (0,1)^{\otimes3})  -  (1,1)^{\otimes3}$. Here each tensor power represents a truth-table of 8 entries, or a vector of  dimension 8;
the linear combination is the truth-table for the
symmetric signature $f = [3,-1,-1,3]$, which is in fact a short hand notation for the vector \[(f_{000}, f_{001}, f_{010}, f_{011},f_{100}, f_{101},f_{110}, f_{111}) = (3, -1, -1, -1, -1, -1, -1, 3).\]
Also note that, $(1,0)H = (1,1)$, $(0,1)H = (1,-1)$ and $(1,1)H = (2,0)$, thus
$[3,-1,-1,3]H^{\otimes 3} =
4 ((1,1)^{\otimes3} + (1,-1)^{\otimes3}) - (2,0)^{\otimes3} = 4[2, 0, 2, 0]-[8, 0, 0, 0] = [0,0,8,0]$, which is equivalent to $[0,0,1,0]$ by a global factor.
So, we get
\begin{equation*}
   \begin{split}
       \holant{[3,-1,-1,3]}{(=_3)} & \equiv_{T} \holant{[3,-1,-1,3]H^{\otimes 3}}{(H^{\otimes 3})^{-1}[1,0,0,1]}\\
       & \equiv_{T} \holant{[0,0,1,0]}{[1,0,1,0]} \\& \equiv_{T} \holant{[0,0,1,0]}{[0,0,1,0]} \\& \equiv_{T} \holant{[0,1,0,0]}{[0,1,0,0]}
   \end{split} 
\end{equation*}
where the first reduction is by Valiant's Holant theorem~\cite{valiant2008holographic},
the third reduction comes from the following observation: given a bipartite 3-regular graph $G = (V,U,E)$ where the vertices in $V$ are assigned the signature $[0,0,1,0]$ and the vertices in $U$ are assigned the signature $[1,0,1,0]$, every nonzero term in the Holant sum must correspond to a mapping $\sigma: E \rightarrow \{0, 1\}$ where exactly two edges of any vertex are assigned 1. The fourth reduction is by simply flipping 0's and 1's. The problem $\holant{[0,1,0,0]}{[0,1,0,0]}$ is the  problem of counting perfect matchings in 3-regular bipartite graphs, which Dagum and Luby proved to be \numP-complete (Theorem 6.2 in~\cite{dagum1992approximating}). 
\end{proof}



\begin{lemma} \label{p3}
The problem $[1, a, -2a-1, 3a+2]$ is \numP-hard unless $a=-1$ in which case it is in {\rm P}.
\end{lemma}
\begin{proof}

Observe that the truth-table of the symmetric signature
$[1, a, -2a-1, 3a+2]$ written as an 8-dimensional column vector is just
{\small
\begin{equation*} \label{magic}
   2(a+1)\left(\left[\begin{array}{c}1 \\ 0\end{array}\right]^{\otimes 3}+\left[\begin{array}{c}0 \\ 1\end{array}\right]^{\otimes 3}\right) - \frac{a+1}{2}\left(\left[\begin{array}{c}1 \\ 1\end{array}\right]^{\otimes 3} +\left[\begin{array}{c}1 \\ -1\end{array}\right]^{\otimes 3}\right)-a\left[\begin{array}{c}1 \\ -1\end{array}\right]^{\otimes 3}.
\end{equation*}
}

Here again, the tensor powers written as 8-dimensional vectors represent  truth-tables,
and the linear combination of these  vectors
``holographically'' reconstitute a truth-table of
 the symmetric signature $[1, a, -2a-1, 3a+2]$.

We apply the holographic transformation with the Hadamard matrix $H = \left[\begin{smallmatrix}1 & 1 \\ 1 & -1\end{smallmatrix}\right]$. Note that  $(1,0)H = (1,1)$, $(0,1)H = (1,-1)$, $(1,1)H = (2,0)$ and $(1,-1)H = (0,2)$, and we get
\begin{equation}
   \begin{split}
       \holant{[1,a,-2a-1,3a+2]}{(=_3)} & \equiv_{T} \holant{[1,a,-2a-1,3a+2]H^{\otimes 3}}{(H^{\otimes 3})^{-1}[1,0,0,1]}\\   & \equiv_{T} \holant{[0,0,a+1,-3a-1]}{[1,0,1,0]}\\& \equiv_{T} \holant{[0,0,a+1,0]}{[0,0,1,0]} 
   \end{split} 
\end{equation}
where the last equivalence follows from the observation that for each nonzero term in the Holant sum,
every vertex on the LHS has at least two  of three
edges  assigned  $1$ (from $[0,0,a+1,-3a-1]$), meanwhile 
every vertex on the RHS has at most two  of three
edges  assigned  $1$  (from $[1,0,1,0]$). The graph being bipartite and 3-regular,
the number of vertices on both sides must equal,
 thus every vertex
has exactly two incident edges assigned 1.

Then by flipping 0's and 1's, $\holant{[0,0,a+1,0]}{[0,0,1,0]} \equiv_{T} \holant{[0,a+1,0,0]}{[0,1,0,0]}$. For $a \ne -1$, this problem is equivalent to counting perfect matchings in  bipartite 3-regular graphs, which is \numP-complete by Theorem 6.2 in~\cite{dagum1992approximating}.
If $a=-1$, the signature $[1, -1, 1, -1] = [1, -1]^{\otimes 3}$ is degenerate, and thus in P. The 
 holographic reduction also reveals that, not only the problem is in P, but 
the Holant sum is $0$.
\end{proof}

We can generalize Lemma~\ref{p3} to get the following corollary.

\begin{corollary}\label{cor:planar}
The problem $\holant{f}{\, =_3}$, where $f= [3a+b,-a-b, -a+b, 3a-b]$, is
computable in polynomial time on planar graphs for all $a, b$, but is \numP-hard on general graphs for all $a\ne 0$.
\end{corollary}
\begin{proof}
The following equivalence is by a holographic transformation
using $H = \left[\begin{smallmatrix}1 & 1 \\ 1 & -1\end{smallmatrix}\right]$:
\begin{eqnarray*}
       \holant{f}{(=_3)} & \equiv_{T} & \holant{fH^{\otimes 3}}{(H^{-1})^{\otimes 3}(=_3)}\\  &  \equiv_{T} & \holant{[0,0, a,b]}{[1,0,1,0]}\\& \equiv_{T} & \holant{[0,0,a,0]}{[0,0,1,0]} \\  &\equiv_{T} & \holant{[0,a,0,0]}{[0,1,0,0]} 
\end{eqnarray*}
where the third reduction follows the same reasoning as in the proof of Lemma~\ref{p3}.
When $a\ne 0$, $\holant{[0,a,0,0]}{[0,1,0,0]}$ is (up to a global nonzero factor) the perfect matching problem on 3-regular bipartite graphs. This problem is computable in polynomial time on planar graphs and the reductions are valid for planar graphs as well. It is \numP-hard on general graphs (for $a\ne 0$).
\end{proof}

\noindent
\textbf{Remark}:
The planar tractability
of the problem $\holant{f}{\,=_3}$, for $f= [3a+b,-a-b, -a+b, 3a-b]$,  is a remarkable fact. It is
neither accomplished by a holographic transformation
to matchgates alone, nor entirely independent from it.
One can prove that the signature 
$f$ is not matchgates-transformable
(for nonzero $a, b$; see  \cite{cai2017complexity} for the theory of matchgates and the realizability of
signatures by matchgates under  holographic transformation).
In previous complexity dichotomies, we have found that for
the entire class of counting CSP problems over Boolean variables, all
problems that are \#P-hard in general, but P-time tractable 
on planar graphs, are  tractable by the following universal algorthmic strategy---a
 holographic transformation
to matchgates followed by the FKT algorithm~\cite{CaiLX10}. On the other hand, for (non-bipartite)
Holant problems with arbitrary symmetric signature sets, this category of problems (planar
tractable but \#P-hard in general) 
is completely characterized by two types~\cite{CaiFGW15} : (1) holographic transformations
to matchgates, and (2) a separate kind that depends on the existence of ``a wheel structure'' (unrelated to 
holographic transformations
and matchgates). Here in Corollary~\ref{cor:planar}
we have found the first instance where a new type has emerged.


\begin{proposition} \label {g1-rou}
For $G_1 = \left[\begin{smallmatrix}1 & b \\ a & c\end{smallmatrix}\right]$, with $a,b,c \in \mathbb{Q}$, if it is non-singular (i.e., $c \ne ab$), then it has two nonzero eigenvalues $\lambda$ and $\mu$. The ratio  $\lambda/\mu$
is not a root of unity \emph{unless} at least one of the following conditions holds: \begin{equation}\label{c1eq} \begin{cases}

c+1=0\\
ab+c^2 + c + 1= 0\\
2ab + c^2 + 1= 0\\
3ab+c^2-c+1=0\\
4ab + c^2-2c+1=0\\
\end{cases}
\end{equation}
\end{proposition}
\begin{proof}
We have $\lambda=\frac{-\Delta+(1+c)}{2}$ and $\mu=\frac{\Delta+(1+c)}{2}$, where $\Delta=\sqrt{(1-c)^2 + 4ab}$. Since $a,b,c \in \mathbb{Q}$,  if
$\frac{\lambda}{\mu}$ is a root of unity it belongs to
 an extension field of $\mathbb{Q}$ of degree 2. Thus it
can  only be
one of the following 8 values: $\pm 1$, $\pm i$, $\frac{\pm 1 \pm \sqrt{3}i}{2}$, where $i=\sqrt{-1}$. 
This gives us the cases listed in (\ref{c1eq}).
\end{proof}

Now we introduce two more binary straddled signatures --- $G_2$ and $G_3$ in Figure \ref{g2g3}.
The signature matrix of 
$G_2$ is $
\left[\begin{smallmatrix}
w & b' \\ a' & c'
\end{smallmatrix}\right]$,
where $w=1+2a^3+b^3$, $a'=a+2a^2b+b^2c$, $b'=a^2+2ab^2+bc^2$ and $c'=a^3+2b^3+c^3$. Similar to $G_1$, we have
$\Delta' = \sqrt{(w-c')^2+4a'b'}$, two eigenvalues $\lambda'=\frac{-\Delta' + (w+c')}{2}$ and $\mu' = \frac{\Delta' + (w+c')}{2}$. If $a'\ne 0$, we have $x'=\frac{\Delta'-(w-c')}{2a'}$, $y'=\frac{\Delta'+(w-c')}{2a'}$ and if further $\Delta' \ne 0$ we can write its Jordan Normal Form as 
\begin{equation}\label{eqn:G2-JNF}
G_2 = \left(\begin{array}{ll}w' & b' \\ a' & c'\end{array}\right)=\left(\begin{array}{cc}-x' & y' \\ 1 & 1\end{array}\right)\left(\begin{array}{ll}\lambda' & 0 \\ 0 & \mu'\end{array}\right)\left(\begin{array}{cc}-x' & y' \\ 1 & 1\end{array}\right)^{-1}.
\end{equation}
The signature matrix of 
 $G_3$ is $\left[\begin{smallmatrix}1+ab & a^2+bc \\ a+b^2 & ab+c^2\end{smallmatrix}\right]$. In this case we define $w=1+ab$, $a'=a+b^2$, $b'=a^2+bc$ and $c'=ab+c^2$. Then the  corresponding quantities 
 $\Delta',  \lambda', \mu', x',  y'$ can be defined
 in the same way, and
 its Jordan Normal Form  takes the same form as
 in (\ref{eqn:G2-JNF}).


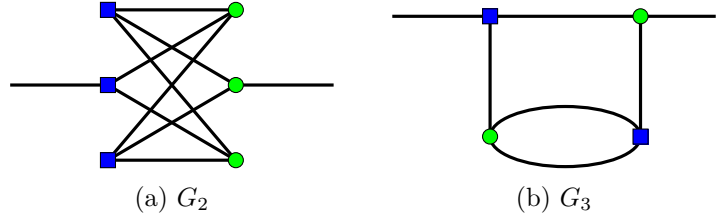
\begin{figure}[h!] 
    \centering
    \begin{subfigure}[b]{0.3\textwidth}
        \centering
        \begin{tikzpicture} 
            \draw[very thick] (0,4.3)--(1.2,4.3);
            \draw[very thick] (3.1,4.3) -- (4.3,4.3);
            \draw[very thick] (1.3,5.3)-- (3,5.3);
            \draw[very thick] (1.3,4.3)-- (3,5.3);
            \draw[very thick] (1.3,3.3)-- (3,5.3);
            \draw[very thick] (1.3,5.3)-- (3,4.3);
            \draw[very thick] (1.3,3.3)-- (3,4.3);
            \draw[very thick] (1.3,5.3)-- (3,3.3);
            \draw[very thick] (1.3,4.3)-- (3,3.3);
            \draw[very thick] (1.3,3.3)-- (3,3.3);
            \filldraw[fill=blue] (1.2,5.2) rectangle(1.4,5.4);
            \filldraw[fill=blue] (1.2,4.2) rectangle(1.4,4.4);
            \filldraw[fill=blue] (1.2,3.2) rectangle(1.4,3.4);
            \filldraw[fill=green] (3,5.3) circle (0.1);
            \filldraw[fill=green] (3,4.3) circle (0.1);
            \filldraw[fill=green] (3,3.3) circle (0.1);
        \end{tikzpicture}
        \caption{$G_2$}
        \label{f2}
    \end{subfigure}
    \begin{subfigure}[b]{0.3\textwidth}
        \centering
        \begin{tikzpicture}  
            \draw[very thick] (0,4.3)--(4.3,4.3);
            \draw[very thick] (1.3,2.7)--(1.3,4.3);
            \draw[very thick] (3.3,2.7)--(3.3,4.3);
            \draw[very thick] (2.3,2.7)ellipse(1 and 0.4);
            \filldraw[fill=blue](1.2,4.2)rectangle(1.4,4.4);
            \filldraw[fill=green] (3.3,4.3)circle(0.1);
            \filldraw[fill=blue] (3.2,2.6)rectangle(3.4,2.8);
            \filldraw[fill=green] (1.3,2.7)circle(0.1);
        \end{tikzpicture}
        \caption{$G_3$}
        \label{f3}
    \end{subfigure}
    \caption{Two binary straddled gadgets}
    \label{g2g3}
\end{figure}

Similar to Proposition \ref{g1-rou}, we have the following  claim on $G_2$ and $G_3$.


\begin{proposition} \label{c3}
For each gadget $G_2$ and $G_3$ respectively, 
if the signature matrix  is non-degenerate, then
 the ratio  $\lambda'/\mu'$ of its eigenvalues 
is not a root of unity \emph{unless} at least one of the following conditions holds, where 
$A = w + c', B = (c'-w)^2 + 4a'b'$.
\begin{equation} \label{c3eq}\begin{cases}
A=0\\
B=0 \\
A^2+B=0\\
A^2+3B=0\\
3A^2 + B = 0\\
\end{cases}
\end{equation}
\end{proposition}

\begin{lemma} \label{y11x}
Suppose  $a,b,c \in \mathbb{Q}$, $a\neq0$ and $c\neq ab$ and $a,b,c$ do not satisfy any condition in  (\ref{c1eq}). Let $x=\frac{\Delta-(1-c)}{2 a}$,  $y=\frac{\Delta+(1-c)}{2 a}$ and $\Delta=\sqrt{(1-c)^2 + 4ab}$. Then 
for $\Holant ( \, [1,a,b,c] \, | =_3)$,
\begin{enumerate}
    \item we can  interpolate $[y,1]$ on LHS;
    \item we can  interpolate $[1,x]$ on RHS except for 2 cases:
    $[1,a,a,1]$, $[1,a,-1-2a, 2+3a]$.
\end{enumerate}  
\end{lemma}
\begin{proof}
This lemma follows from  Lemma \ref{3.1.1} and Lemma \ref{3.1.2}
using the binary straddled gadget $G_1$ with singaure matrix $\left[\begin{smallmatrix}1 & b \\ a & c\end{smallmatrix}\right]$.  Note that $c\neq ab$ indicates that matrix $G_1$ is non-degenerate, and $\lambda/\mu$ not being a root of unity is equivalent to none of the equations in (\ref{c1eq}) holds.
\end{proof}

We have similar statements corresponding to $G_2$ (resp. $G_3$). When the signature matrix is non-degenerate and does not satisfy any condition in  (\ref{c3eq}), we can  interpolate the corresponding $[y',1]$ on LHS, and we can also interpolate the corresponding $[1,x']$ on RHS except when $y'=-1$.


\begin{definition} \label{d1}
For $\Holant ( \, [1,a,b,c] \, | =_3)$, with $a,b,c \in \mathbb{Q}$, $a \neq 0$, we say a binary straddled gadget $G$ \emph{works} if the signature matrix of $G$ is non-degenerate and the ratio of its two eigenvalues $\rlm$ is not a root of unity.
\end{definition}
\noindent
\textbf{Remark}:  Explicitly, the condition that $G_1$ \emph{works} is that $c \neq ab$ and $a,b,c$ do not satisfy any condition in (\ref{c1eq}), which is just the assumptions in Lemma \ref{y11x}. $G_1$ \emph{works} implies that it can be used to interpolate  $[y,1]$ on LHS, 
and to interpolate $[1,x]$ on RHS with two exceptions 
for which we already proved the dichotomy. The $x,y$ are as stated in Lemma \ref{y11x}.

Similarly, when the binary straddled gadget $G_2$ (resp. $G_3$) \emph{works}, for the corresponding values $x'$ and $y'$,
we can  interpolate $[y',1]$ on LHS,
and we can interpolate  $[1,x']$ on RHS except when $y'=-1$.

\begin{figure}[h!] 
\centering
\begin{tikzpicture}  
\draw[very thick] (0.4,7.3)--(6.2,7.3);
\draw[very thick] (1.3,7.3)--(3.3,3.8);
\draw[very thick] (5.3,7.3)--(3.3, 3.8);
\draw[very thick] (3.3, 3.8)--(3.3, 2.9);
\draw[very thick] (3.3,6.1)--(3.3,7.3);
\draw[very thick] (3.3,6.1)--(4.3, 5.55);
\draw[very thick] (3.3,6.1)--(2.3,5.55);
\filldraw[fill=blue] (1.2,7.2) rectangle(1.4,7.4);  
\filldraw[fill=blue] (5.2,7.2) rectangle(5.4,7.4);  
\filldraw[fill=blue] (3.2,3.7) rectangle(3.4,3.9); 
\filldraw[fill=green] (3.3,7.3)circle(0.1);
\filldraw[fill=green] (4.3,5.55)circle(0.1);
\filldraw[fill=green] (2.3,5.55)circle(0.1);
\filldraw[fill=blue] (3.2, 6.0) rectangle(3.4, 6.2);     
\end{tikzpicture}
  \caption{$G_4$}
  \label{f4}
\end{figure}

\vspace{.1in}

The ternary gadget $G_4$ in Figure \ref{f4} will be used in the proof here and later.

The unary signatures  $\Delta_{0}=[1,0]$ and $\Delta_{1}=[0,1]$
are called  the pinning signatures because they ``pin''
a variable to 0 or 1.
One good use of having unary signatures is that we can use Lemma~\ref{3.2} to get the two 
pinning signatures. Pinning signatures
are helpful as the following lemma shows.

\begin{lemma} \label{after-unary}
If $\Delta_{0}$ and $\Delta_{1}$ are available on the RHS in
$\Holant ( \, [1,a,b,c] \, | =_3)$, 
where $a,b,c \in \mathbb{Q}$, $ab \neq 0$,
then  the problem is \numP-hard unless $[1,a,b,c]$ is affine or degenerate, in which cases it is in P.
\end{lemma}
\begin{proof}
Connecting 
$[1,0]$, $[0,1]$ to $[1,a,b,c]$ on LHS respectively, we  get binary signatures $[1,a,b]$ and $[a,b,c]$. Then we can apply Theorem \ref{2-3},
and the problem is \numP-hard unless both $[1,a,b]$ and $[a,b,c]$ are in P. When $ab\neq 0$, both $[1,a,b]$ and $[a,b,c]$ are in P only when $[1,a,b,c]=$  $[1,1,1-1]$ or $[1,-1,1,1]$ or $[1,-1,-1,-1]$ or $[1,1,-1,1]$ or $[1,1,-1,-1]$ or $[1,-1,-1,1]$, where the last two are affine and hence in P. Due to the symmetry by flipping 0 and 1 in the  signature, it suffices to consider only $f= [1,1,1,-1]$ and $g=[1,-1,1,1]$; they   are neither affine nor degenerate.

For both  $f$ and $g$ we use the gadget $G_4$ to produce  ternary signatures  $f' = [1,1,3,3]$
and $g' = [1,1,-1,3]$ respectively.  Neither are among  the exceptional cases above.
So $\Holant ( \, f \, | =_3)$ and $\Holant ( \, g\, | =_3)$ are both \numP-hard.
\end{proof}

\vspace{.1in}

The following lemma lets us interpolate arbitrary unary signatures on RHS, in particular  $\Delta_{0}$ and $\Delta_{1}$, from a binary gadget with a straddled signature and a suitable unary signature $s$ on RHS.  Mathematically, the proof is essentially the same as in \cite{vadhan2001complexity}, but technically Lemma~\ref{3.2} applies to  binary straddled signatures. 

\begin{lemma} \label{3.2}
Let $M\in \mathbb{R}^{2\times 2}$ be a non-singular signature matrix for a binary straddled gadget which is diagonalizable with distinct eigenvalues, and $s=[a,b]$ be a unary signature on RHS that is not a row eigenvector of $M$. Then $\{s\cdot M^j\}_{j\geq 0}$ can be used to interpolate any unary signature on RHS.
\end{lemma}

\section{Dichotomy when $ab\neq 0$ and $G_1$ works}  \label{large2}

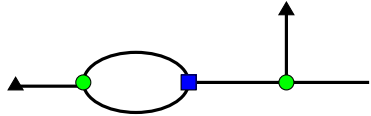
\begin{figure}[h!] 
\centering
\begin{tikzpicture}  
\draw[very thick] (2,3.3)ellipse(0.7 and 0.4);
\filldraw[fill=black] (0.4,3.37)--(0.3,3.2)--(0.5,3.2)--cycle; 
\draw[very thick] (0.34,3.25)--(1.3,3.25);
\draw[very thick] (2.7,3.3)--(5.1,3.3);
\draw[very thick] (4,4.3)--(4,3.3);
\filldraw[fill=green] (1.3, 3.3)circle(0.1);  
\filldraw[fill=blue] (2.6,3.2) rectangle(2.8,3.4);  
\filldraw[fill=green] (4,3.3)circle(0.1);
\filldraw[fill=black] (4,4.37)--(3.9,4.2)--(4.1,4.2)--cycle;
\end{tikzpicture}

  \caption{Non-linearity gadget, where a triangle represents the unary gadget $[y,1]$}
  \label{non_li}
\end{figure}

Let us introduce a \emph{non-linearity} gadget in Figure \ref{non_li} where the triangles represent $[y,1]$. It is on
the RHS with a unary signature $[y^2+yb, ya+c]$. 
%
The following two lemmas  will be used in proof of  Theorem~\ref{g1works1}.

\begin{lemma} \label{1-1 X=1}
Let $a,b, c \in \mathbb{Q}$, $ab \neq 0$,  and
satisfy \emph{(con1)} $a^3 - b^3 - ab(1-c) =0$ and
\emph{(con2)} $a^3+ab+2b^3=0$.
Then $\Holant ( \, [1,a,b,c] \, | =_3)$ is
\numP-hard unless it is $[1,-1,1,-1] = [1, -1]^{\otimes 3}$, which is degenerate, in which case the problem is in P. 
\end{lemma}
\begin{proof}
If $a+b^2=0$, then $[1,a,b,c] = [1,-1,1,-1]$ which is degenerate. Now we assume $a+b^2\ne0$. Here we use Gadget $G_3$. 

First assume $G_3$ \emph{works}. 
Using  $a+b^2\ne 0$
together with (\emph{con1}) and  (\emph{con2}), we 
can verify that 
$\Delta = \sqrt{4(a+b^2)(a^2+bc)+(c^2-1)^2} \not =0$, and
we can write the Jordan Normal Form  $$G_3 = \left(\begin{array}{ll}1+ab & a^2+bc \\ a+b^2 & ab+c^2
\end{array}\right)=\left(\begin{array}{cc}-x & y \\ 1 & 1\end{array}\right)\left(\begin{array}{ll}\lambda & 0 \\ 0 & \mu\end{array}\right)\left(\begin{array}{cc}-x & y \\ 1 & 1\end{array}\right)^{-1},$$ where
$\lambda = \frac{1+2ab+c^2-\Delta}{2}$, $\mu = \frac{1+2ab+c^2+\Delta}{2}$, $x = \frac{\Delta+c^2-1}{2(a+b^2)}$, $y = \frac{\Delta-c^2+1}{2(a+b^2)}$. Because $G_3$ works, $[y,1]$ on LHS is  available. Use this $[y,1]$
in the non-linearity gadget in Figure \ref{non_li}, we get the unary signature $\left[y^{2}+y b, y a+c\right]$ on the RHS. By Lemma \ref{3.2}, we can interpolate any unary signature, in particular $\Delta_{0}$ and $\Delta_{1}$ on RHS
and apply Lemma~\ref{after-unary}, unless $\left[y^{2}+y b, y a+c\right]$ is proportional to a row eigenvector of $G_3$, namely $[1,-y]$ and $[1, x]$. Thus the exceptions are $ya+c = x(y^{2}+y b)$ and $ya+c = -y(y^2+yb)$. Notice that now $xy = \frac{a^2+bc}{a+b^2}$. The first equation implies $c = ab$ or $a + b^2= 0$ or $a^3-b^3c+ab(-1+c^2) = 0$. The second equation implies $a+b^2=0$ or $f_1 =0$ where $f_1= a^3+4a^6+3a^5b^2+a^3b^3-c-4a^3c+6a^4bc-6a^2b^2c-b^3c-3a^2b^5c-3a^3c^2-3abc^2-4b^3c2-a^3b^3c^2-6ab^4c^2-4b^6c^2+3c^3+4a^3c^3+6a^2b^2c^3+3b^3c^3+a^3c^4+3abc^4+4b^3c^4-3c^5-b^3c^5+c^7$. So there are four exceptional cases, \begin{equation} \label {g3works} \begin{cases}
    c = ab \\
    a + b^2= 0\\
    a^3-b^3c+ab(-1+c^2) = 0\\
    f_1 = 0\\ \end{cases}
\end{equation}

 For each of them, together with (\emph{con1}) and (\emph{con2}), we get 3 equations and can solve them using
Mathematica\texttrademark{}.
For rational $a,b,c$, when $ab\neq0$, there are only two possible results --- $[1,-1,1,-1]$ and $[1,-\frac{1}{3}, -\frac{1}{3}, 1]$. The first one violates $a+b^2 \ne 0$, 
and the second has been proved to be \numP-hard in Lemma \ref{easy}. For all other cases when $G_3$ works, we have
the pinnng signatures $\Delta_{0}$ and $\Delta_{1}$ 
on the RHS and then the lemma is proved by Lemma \ref{after-unary}.

Now suppose $G_3$ does not work. Then by Proposition \ref{c3}, we get at least one more condition, either one in (\ref{c3eq}) or $(a+b^2)(a^2+bc) = (1+ab)(ab+c^2) $ which indicates that $G_3$ is degenerate. For each of the 6 conditions, together with (\emph{con1}) and (\emph{con2}), we can solve them using
Mathematica\texttrademark{} for rational $a,b,c$. The only 
solution is $[1,-1,1,-1]$ which  violates $a+b^2 \ne 0$.
The proof of the  lemma is complete.
\end{proof}

\begin{lemma} \label{1-2 X=1}
Let $a,b, c \in \mathbb{Q}$, $ab \neq 0$,  and
satisfy \emph{(con1)} $a^3 - b^3 - ab(1-c) =0$ and \emph{(con2)}  $(a^4b+ab^4)^2 = (a^5+b^4)(b^5+a^4c)$.
Then $\Holant ( \, [1,a,b,c] \, | =_3)$ is
 \numP-hard unless it is $[1,a,a^2, a^3]$, which is degenerate and thus in P. 
\end{lemma}
\begin{proof}
Eliminating $c$ from  (\emph{con1}) and (\emph{con2})
we get 
 $a^{11}-a^9b+a^6b^4+a^5b^6-a^4b^5-a^3b^7+a^2b^9-b^{10}=0$, which, quite miraculously, can be factored as $(a^2-b) (a^9 + a^4 b^4 + a^3b^6 + b^9)=0$. If $b=a^2$, then with (\emph{con1}), we get $c=a^3$ thus the signature becomes $[1,a,a^2,a^3]$ which is degenerate. We  assume 
 $a^9 + a^4b^4 + a^3b^6+b^9=0$, then the rest  of the
 proof is essentially the same as Lemma \ref{1-1 X=1}.
%
%
%
\end{proof}

\begin{theorem} \label{g1works1}
For $a,b, c \in \mathbb{Q}$, $ab \neq 0$, 
 if $G_1$ works, then the problem $\Holant ( \, [1,a,b,c] \, | =_3)$ is
 \numP-hard unless it is degenerate or Gen-Eq or  affine, and thus in P.
\end{theorem}

\begin{proof}
If $[1,a,b,c]$ has the form $[1,a,a,1]$ or $[1,a,-1-2a,2+3a]$
then the dichotomy has been proved in Lemmas \ref{easy} and  \ref{p3} respectively.  We now assume the signature is not of these two forms.
By Lemma \ref{y11x}, when $G_1$ works, we can interpolate $[y,1]$ on LHS and also  $[1,x]$ on RHS.

Let us write down the Jordan Normal Form again: $$G_1 = \left(\begin{array}{ll}1 & b \\ a & c\end{array}\right)=\left(\begin{array}{cc}-x & y \\ 1 & 1\end{array}\right)\left(\begin{array}{ll}\lambda & 0 \\ 0 & \mu\end{array}\right)\left(\begin{array}{cc}-x & y \\ 1 & 1\end{array}\right)^{-1},$$ and $\lambda=\frac{-\Delta+(1+c)}{2}$, $\mu=\frac{\Delta+(1+c)}{2}$, $x=\frac{\Delta-(1-c)}{2 a}$,  $y=\frac{\Delta+(1-c)}{2 a}$, $\Delta=\sqrt{(1-c)^2 + 4ab}$.

Using $[y,1]$ and the gadget in Figure \ref{non_li}, we get  $\left[y^{2}+y b, y a+c\right]$ on the RHS.
We can interpolate  $\Delta_{0}$ and $\Delta_{1}$ on RHS unless $\left[y^{2}+y b, y a+c\right]$ is proportional to a row eigenvector of $G_1$, namely $[1,-y]$ or $[1, x]$, according to Lemma \ref{3.2}. Thus the exceptions are $y a+c = (y^{2}+y b)x$ or $y a+c = -y(y^{2}+y b)$. The first equation implies $a^{3}-b^{3}-a b(1-c)=0$ or $c=ab$. The second equation implies $c=ab$ or $c=1+a-b$. 

By assumption  $G_1$ works, so $c\ne ab$. Thus, we consider two exceptional cases.

\noindent 
\textbf{Case 1}: $a^{3}-b^{3}-a b(1-c)=0$

In this case, we have $1-c = \frac{a^3-b^3}{ab}$ and thus $\Delta = \sqrt{(1-c)^2+4ab} = |\frac{a^3+b^3}{ab}|$.  One condition ($4ab + c^2 - 2c+1 = 0$) in (\ref{c1eq}) is the same as $\Delta=0$.  Since $G_1$ works, we have $\Delta \neq 0$ and thus $a^3+b^3\neq 0$, which is equivalent to $a+b\ne 0$ when $a,b\in \mathbb{Q}$.

\begin{description}
    \item{Subcase 1: $ \frac{a^3+b^3}{ab} > 0$.}
%
 We have  $[1,x] = [1,\frac{\Delta - \left(1-c\right)}{2a}] = [1, \frac{b^2}{a^2}]$ on RHS. Connect $[1,x]$ to $[1,a,b,c]$ on LHS, we get the binary signature $[1+\frac{b^2}{a}, a+\frac{b^3}{a^2}, b+\frac{b^2c}{a^2}]$ on LHS. 
%
It is \numP-hard (and thus the problem $[1,a,b,c]$ is \numP-hard) unless one of the tractable conditions in Theorem~\ref{2-3} holds. It turns out that the only possibility is $X=1$
in Theorem \ref{2-3}, which becomes $\left(a^2-b \right)\left(a^3+ab+2b^3 \right) = 0$. 
When $a^2-b = 0$, together with $a^3-b^3=ab\left(1-c\right)$, we have $c = a^3$, and thus $[1,a,b,c]$ is degenerate. When $a^3+ab+2b^3 = 0$, together with $a^3-b^3=ab\left(1-c\right)$, by Lemma \ref{1-1 X=1}, $[1,a,b,c]$ is \numP-hard (with $a+b\ne0$ ruling out the exception).

 \item{Subcase 2: $\frac{a^3+b^3}{ab} < 0$}. We have $[y,1] = [-\frac{b^2}{a^2},1]$ on LHS. Connecting two copies of $[y,1]$ to $(=_3)$ we get $[y^2, 1] = [\frac{b^4}{a^4}, 1]$ on RHS.
 Connecting it back to LHS, we  get a binary signature $[\frac{b^4}{a^4}+a,\frac{b^4}{a^3}+b,\frac{b^5}{a^4}+c]$ on LHS. 
It is \numP-hard  unless one of the tractable conditions in Theorem~\ref{2-3} holds. It turns out that the only possibility is $X=1$
in Theorem \ref{2-3}, which becomes 
 $(a^4b+ab^4)^2 = (a^5+b^4)(b^5+a^4c)$. Together with $a^3-b^3=ab\left(1-c\right)$, by Lemma \ref{1-2 X=1}, $[1,a,b,c]$ is \numP-hard unless it is degenerate.
\end{description}

\noindent 
\textbf{Case 2}: $1+a-b-c=0$

In this case, $\Delta=|a+b|$, and since
$G_1$ works, one condition is $4ab+c^2-2c+1=0$ in (\ref{c1eq})
which says $\Delta \ne 0$,  and thus $a+b \neq 0$.

If $a+b > 0$, then $x = \frac{a+b-(1-c)}{2a} = 1$. Then we can interpolate $[1,x] = [1,1]$ on RHS (as $y = \frac{a+b+(1-c)}{2a} = \frac{b}{a}\ne -1$). Else,  $a+b < 0$, $y = \frac{-a-b + (1-c)}{2a} = \frac{-a-b + (b-a)}{2a}= -1$. We can  get $[1,1]$ on RHS by connecting two copies of $[y,1]=[-1,1]$ to $[1,0,0,1]$. Then connecting $[1,1]$ to $[1,a,b,c]$ on LHS we get a binary signature $[a+1,a+b, a+1]$ on LHS. Again we can apply
 Theorem \ref{2-3} to it, and conclude that it is
 \numP-hard. It turns out that  
 the only feasible case of tractability is $X=1$ in  Theorem \ref{2-3}, which leads to 
  $[1,a,-1-2a,2+3a]$, but we assumed $[1,a,b,c]$ is not of this form.
This proves the \numP-hardness of
$\Holant ( \, [1,a,b,c] \, | =_3)$.
%
%
%
%
\end{proof}

\section{Dichotomy for $[1,a,b,0]$}   \label{large3}

\begin{theorem}   \label{1ab0}
The problem  $[1,a,b,0]$ 
for $a,b \in \mathbb{Q}$ is
 \numP-hard
 unless it is degenerate 
 or  affine, and thus in P.
\end{theorem}

\begin{proof}
If $ab\neq 0$ and $G_1$ works, then this is proved in
Theorem~\ref{g1works1} (in this case, it cannot be a Gen-Eq).
 If $a=b=0$, it is degenerate and in P. We divide the rest into three cases: \begin{enumerate}
    \item $ab\neq 0$ and $G_1$ does not work;
    \item $f = [1,a,0,0]$ with $a\neq 0$;
    \item $f= [1,0,b,0]$ with $b\neq 0$.
\end{enumerate}

\noindent 
$\bullet$ Case 1:  $ab\neq 0$ in $f = [1,a,b,0]$ and $G_1$ does not work. Since $c=0 \ne ab$, this implies that at least one equation in (\ref{c1eq})
holds. After a simple derivation, we  have the following family of signatures to consider:   $[1,a,-\frac{1}{ka}, 0]$, for $k = 1, 2, 3, 4$. 

We  use $G_4$ to produce another symmetric ternary signature in each case. If the new signature is \numP-hard, then so is the given signature. We will describe the case $[1,a,-\frac{1}{a}, 0]$ in more detail; the other three types ($ k=2, 3, 4$) are similar.

For $k=1$, the gadget $G_4$ produces $g= [3 a^3+4, a^4-a-\frac{2}{a^2}, -a^2+\frac{1}{a}+\frac{1}{a^4}, a^3+3]$. 
For $a= -1$, this is $[1,0,-1,2]$, which has the form $[1,a',-1-2a', 2+3a']$ and is \numP-hard by Lemma \ref{p3}. Below we assume $a \not = -1$. Then all entries of $g$ are nonzero.

We claim that the gadget $G_1$ works using $g$. Since $a\in \mathbb{Q}$, it can be checked that $g$ is non-degenerate since $(a^4-a-\frac{2}{a^2})( -a^2+\frac{1}{a}+\frac{1}{a^4}) = (3a^3 + 4)(a^3 + 3)$ has no solution, and that no equation in (\ref{c1eq}) has a solution applied to $g$. Hence, $G_1$ works using $g$ and we may apply Theorem \ref{g1works1} to $g$. Using the fact that $a\in\mathbb{Q}$, one can show that $g$ cannot be a Gen-Eq because it has no zero entry, nor can it be affine or degenerate. Thus $[1,a,-\frac{1}{a}, 0]$ is \numP-hard.

\noindent 
$\bullet$ Case 2: $f= [1,a,0,0]$ with $a\neq 0$. 
The gadget $G_4$ produces $g'=[3a^3+1, a^4+a, a^2, a^3]$. Since $a\in \mathbb{Q}$, $3a^3+1\neq0$. 
If $a=-1$, $g'=[-2, 0,1,-1]$ and it suffices to consider $[1,-1,0,2]$, in which case $G_3$ works where the matrix $G_3=\left[\begin{smallmatrix}1 & 1 \\ -1 & 4\end{smallmatrix}\right]$. We can interpolate $[1,x] = [1, -\frac{3 + \sqrt{5}}{2}]$ on RHS. Connect it back to $[1,-1,0,2]$ and get a binary signature $[\frac{5+\sqrt{5}}{2}, -1, -(3+\sqrt{5})]$ on LHS, which, by Theorem \ref{2-3}, is \numP-hard. Thus, $[1,-1, 0,2]$ is \numP-hard and so is $[1,a,0,0]$.

Else, $a\neq -1$.
We claim that the gadget $G_1$ works using $g'$.
The signature $g'$ is non-degenerate 
since $a\in \mathbb{Q}$  is nonzero
and thus  $(a^4+a)a^2 \not =(3a^3+1)a^3$.
Also 
 no equation in (\ref{c1eq}) has a solution applied to $g'$. Hence, $G_1$ works using $g'$ and we may apply Theorem \ref{g1works1} to $g'$. Using the fact that $a\in\mathbb{Q}$, one can show that $g'$ cannot be a Gen-Eq because it has no zero entry, nor can it be affine or degenerate. Thus $[1,a,0, 0]$ is \numP-hard.


\noindent 
$\bullet$ Case 3: 
 $f = [1,0,b,0]$ with $b\neq 0$. The gadget $G_1$ produces a binary straddled signature $G_1=\left[\begin{smallmatrix}1 & b \\ 0 & 0\end{smallmatrix}\right] = \left[\begin{smallmatrix} 1 \\ 0 \end{smallmatrix}\right]\cdot \left[\begin{smallmatrix}1 & b \end{smallmatrix}\right]$ 
which decomposes into 
a unary signature $[1,b]$ on RHS and 
a unary signature $[1, 0]$ on LHS.
This gives us a reduction
$\Holant ( f  | \{ ( =_3), [1,b]\}) \le_T
\Holant ( \, f \, | =_3)$. To see that,
notice that in any signature grid for the
problem $\Holant ( f  | \{ ( =_3), [1,b]\})$,
the number  of occurrences of $[1,b]$ is 0 mod 3, say  $3n$.
We can replace each occurrence $[1,b]$ by $G_1$, leaving
 $3n$
extra  copies  of   $[1,0]$ on the LHS. These can all be
absorbed by connecting to
$(=_3)$. 

Now, if we connect $[1,b]$ to 
 $[1,0,b,0]$ and  get a binary signature $[1,b^2, b]$ on LHS.
Thus, $\Holant ( [1,b^2, b]  | =_3) \le_T
\Holant ( f  | \{ ( =_3), [1,b]\}) $.
 The problem 
$\Holant ( [1,b^2, b]  | =_3)$ is \numP-hard except $b =\pm 1$,
by Theorem \ref{2-3}, which implies that $\Holant ( \, f \, | =_3)$ is also \numP-hard. If $b = \pm 1$,
then $f$ is affine, and $\Holant ( \, f \, | =_3)$ is in P.
%
%
%
\end{proof}

\section{Dichotomy for $[1,a,0,c]$}  \label{large4}
\begin{theorem} \label{1a0c}
The problem $[1,a,0,c]$ with $a,c\in \mathbb{Q}$ is \numP-hard unless $a=0$, in which case it is Gen-Eq and thus in P.
\end{theorem}
\begin{proof}
When $a=0$, it is Gen-Eq and so is in P. When $a\neq 0$, if $c=0$, it is \numP-hard by Theorem \ref{1ab0}. In the following we discuss $[1,a,0,c]$ with $ac\neq 0$.

 If $c=\pm1$, the signature is $[1,a,0,1]$ or $[1,a,0,-1]$. We use $G_4$ to produce a ternary signature $g=[3a^3+1, a^4+a,, a^2, a^3 + 1]$ (both mapped to the same signature, surprisingly). If $a=-1$, it is $[1,0,-\frac{1}{2}, 0]$ after normalization, which by Theorem \ref{1ab0} is \numP-hard and so is the given signature $[1,-1,0,1]$. If $a\neq -1$, then $g$ has no zero entry. We then claim that the gadget $G_1$ works using $g$. It can be checked that $g$ is non-degenerate since $(a^4+a)a^2=(3a^3+1)(a^3+1)$ has no solution, and that no equation in (\ref{c1eq}) has a solution applied to $g$. Hence, $G_1$ works using $g$ and we may apply Theorem \ref{g1works1} to $g$. Using the fact that $a\in\mathbb{Q}$, one can show that $g$ cannot be a Gen-Eq because it has no zero entry, nor can it be affine or degenerate. Thus $[1,a,0,\pm 1]$ are both \numP-hard.


Now assume $c \ne 0, \pm 1$. We claim that the gadget $G_1$ works. It can be checked that for the non-degenerate matrix $G_1=\smm{1}{0}{a}{c}$, $\Delta = |1-c|$, $\rlm \in \{c, \frac{1}{c}\}$  is not a root of unity. 
Next we claim that we can obtain $[1,0]$ on RHS.
If $c<1$ by Lemma~\ref{y11x}
we can interpolate $[1,x]=[1,0]$  on RHS 
with two exceptions to which we already give a dichotomy (see
the Remark after Definition \ref{d1}). 
If $c>1$,  we can interpolate $[y,1] = [0,1]$ on LHS and so the gadget in Figure \ref{non_li} produces  $[0,c]$ on RHS, which is not
proportional to the 
 row eigenvectors $[1,-y] = [1,0]$ and 
 $[1,x] = [1,\frac{c-1}{a}]$ of  $G_1$. 
 By Lemma \ref{3.2}, we can interpolate any unary gadget on RHS, including $[1,0]$. Thus we can always get $[1,0]$ on RHS. Connect $[1,0]$ to $[1,a,0,c]$  and we will get a binary signature $[1,a,0]$ on LHS, which is \numP-hard by Theorem \ref{2-3}. Therefore $[1,a,0,c]$ is \numP-hard when $c\ne 0, \pm 1$. 

\end{proof}

\section{Dichotomy when $abc\ne0$}

We need four lemmas to handle some special cases.

\begin{lemma} \label{p1}
The problem $[1,-b^2, b, -b^3]$ with $b\in \mathbb{Q}$ is \numP-hard unless $b=0, \pm 1$, which is in P.
\end{lemma}
\begin{proof}
If $b=0, \pm 1$, it is degenerate or affine. Now assume $b\neq 0, \pm 1$. $G_1 = \smm{1}{b}{-b^2}{b^3} = \smmv{1}{-b^2}\cdot\smmh{1}{b}$. Then we can get $[1,-b^2]$ on the  LHS similar to the proof in Case 3 of Theorem~\ref{1ab0}.
Note that connecting three copies of $[1,b]$ with $[1,-b^2,b,-b^3]$ on LHS produces a global factor 
$1-b^6 \ne 0$.  Connect $[1,-b^2]$ twice to $[1,0,0,1]$ on RHS, and we get $[1,b^4]$ on RHS. Connect $[1,b^4]$ back to $[1,-b^2, b, -b^3]$ on LHS, and we get a binary signature $g=[1-b^6, -b^2 + b^5, b-b^7]$, which by Theorem \ref{2-3} is \numP-hard, and so is $[1,-b^2, b, -b^3]$. 
\end{proof}

\begin{lemma} \label{p2}
The problem $[1, a, -\frac{1}{a}, -1]$ with $a\in \mathbb{Q}$, $a\neq 0$ is \numP-hard unless $a=\pm 1$, in which case it is in P.
\end{lemma}
\begin{proof}
If $a=\pm 1$, $[1,1,-1-1]$ is   affine and $[1,-1,1,-1]$ is degenerate, both of which are in P. Now we assume $a\neq \pm 1$ (so the matrix $\smm{a^{-2}}{a}{1}{1}$ is invertible). We use the binary straddled gadget $G_2$ and write down its Jordan Normal Form as $$G_2 =\left(\begin{array}{cc} 2a^3 + 1 - a^{-3}  & -a-a^{-2} \\  a^2 + a^{-1}  &   a^3 -1-2a^{-3} \end{array}\right) =\left(\begin{array}{cc}a^{-2} & a \\ 1 & 1\end{array}\right)\left(\begin{array}{ll} a^3 - a^{-3}  & 0 \\ 0 &  2a^3 - 2a^{-3}\end{array}\right)\left(\begin{array}{cc}{a^{-2}} & a \\ 1 & 1\end{array}\right)^{-1}$$ The matrix is non-degenerate and the ratio of its two eigenvalues are $1/2$, so gadget $G_2$ \emph{works}. Since here $y= a \ne \pm1$, we can  interpolate $[1,x] = [1,-{a^{-2}}]$ on RHS. Now connect $[1,x]$  to $[1, a, -\frac{1}{a}, -1]$ on LHS and we can get a binary signature $[1-\frac{1}{a}, a+\frac{1}{a^3}, -\frac{1}{a} + \frac{1}{a^2}]$, which by Theorem \ref{2-3} is \numP-hard, and so is $[1, a, -\frac{1}{a}, -1]$.
\end{proof}

Before presenting our next lemma, we introduce a ternary gadget $G_{aux}$ in Figure \ref{fa}. Its signature is $[1 + 2a^3 + b^3,a + 2a^2b + b^2c,a^2 + 2ab^2 + bc^2, a^3 + 2b^3 + c^3]$. \begin{figure}[h!] 
    \centering
    \begin{tikzpicture}
        \draw[very thick] (1.5,3.5)--(3,4.3);
        \draw[very thick] (1.5,3.5)--(3,3);
        \draw[very thick] (1.5,3.5)--(3,1.7);
        \draw[very thick] (1.5,2.5)--(3,4.3);
        \draw[very thick] (1.5,2.5)--(3,3);
        \draw[very thick] (1.5,2.5)--(3,1.7);
        \draw[very thick] (3,4.3)--(4, 4.3);
        \draw[very thick] (3,3)--(4, 3);
        \draw[very thick] (3,1.7)--(4,1.7);
        \filldraw[fill=green] (1.5,3.5) circle (0.1);  
        \filldraw[fill=green] (1.5,2.5) circle (0.1);  
        \filldraw[fill=blue] (2.9,4.2) rectangle (3.1,4.4); 
        \filldraw[fill=blue] (2.9,2.9) rectangle (3.1,3.1); 
        \filldraw[fill=blue] (2.9,1.6) rectangle (3.1,1.8); 
    \end{tikzpicture}
  \caption{$G_{aux}$}
  \label{fa}
\end{figure}

\begin{lemma} \label{c=ab}
The problem $[1,a,b,ab]$ with $a,b\in \mathbb{Q}$ and $a,b\ne 0$ is \numP-hard unless it is degenerate or  affine, which is in P.
\end{lemma}
\begin{proof}
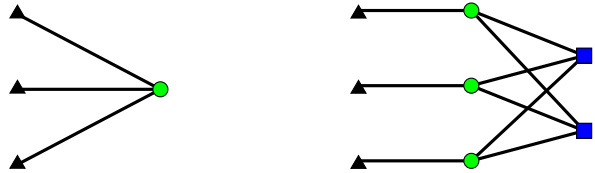
\begin{figure}[h!] 
    \centering
    \begin{subfigure}[b]{0.3\textwidth}
        \centering
        \begin{tikzpicture}  
            \draw[very thick] (0.3,5.05) -- (2.2,4.05);
            \draw[very thick] (0.3,4.05) -- (2.2,4.05);
            \draw[very thick] (0.3,3.05) -- (2.2,4.05);
            \filldraw[fill=black] (0.3,5.17)--(0.2,5)--(0.4,5)--cycle;
            \filldraw[fill=black] (0.3,4.17)--(0.2,4)--(0.4,4)--cycle;
            \filldraw[fill=black] (0.3,3.17)--(0.2,3)--(0.4,3)--cycle;
            \filldraw[fill=green] (2.2, 4.05)circle(0.1);
        \end{tikzpicture}
        \label{i1}
    \end{subfigure}
    \begin{subfigure}[b]{0.3\textwidth}
        \centering
        \begin{tikzpicture}  
            \draw[very thick] (0.3,5.2)--(1.8,5.2);
            \draw[very thick] (0.3,4.2)--(1.8,4.2);
            \draw[very thick] (0.3,3.2)--(1.8,3.2);
            \draw[very thick] (1.8,5.2)--(3.3,4.6);
            \draw[very thick] (1.8,4.2)--(3.3,4.6);
            \draw[very thick] (1.8,3.2)--(3.3,4.6);
            \draw[very thick] (1.8,5.2)--(3.3,3.6);
            \draw[very thick] (1.8,4.2)--(3.3,3.6);
            \draw[very thick] (1.8,3.2)--(3.3,3.6);
            \filldraw[fill=blue] (3.2,4.5)rectangle(3.4,4.7);
            \filldraw[fill=blue] (3.2,3.5)rectangle(3.4,3.7);
            \filldraw[fill=black] (0.3,5.27)--(0.2,5.1)--(0.4,5.1)--cycle;
            \filldraw[fill=black] (0.3,4.27)--(0.2,4.1)--(0.4,4.1)--cycle;
            \filldraw[fill=black] (0.3,3.27)--(0.2,3.1)--(0.4,3.1)--cycle;
            \filldraw[fill=green] (1.8,5.2)circle(0.1);
            \filldraw[fill=green] (1.8,4.2)circle(0.1);
            \filldraw[fill=green] (1.8,3.2)circle(0.1);
        \end{tikzpicture}
        \label{i2}
    \end{subfigure}
    \caption{Two gadgets where each triangle represents the unary gadget $[1,a]$}
    \label{2g41a}
\end{figure}

Using gadget $G_1$, we have a degenerate matrix   $G_1 = \smm{1}{b}{a}{ab} =  \smmv{1}{a}\cdot \smmh{1}{b}$. We  get $[1,b]$ on RHS if $[1,a]$ can appropriately form some nonzero glabal factor. Figure \ref{2g41a} indicates two different ways of ``absorbing"  $[1,a]$ on LHS. The factors they provide are $1+a^3$ and $1+3a^3 + 3a^2b^2+a^5b^2$ respectively. It is easy to see that at least one of them is nonzero except $a=-1$ and $b= \pm 1$, i.e., $[1,-1,1,-1]$ which is degenerate or $[1,-1,-1,1]$ which is affine. Now assume below $[1,a,b,ab]$
is not these two, then we can interpolate $[1,b]$ on RHS. Connect $[1,b]$ back to $[1,a,b,ab]$ on LHS and we get the binary signature $g=[1+ab, a+b^2, b+bc]$. If $a + b^2 = 0$, since $c = ab$, the given signature is $[1, -b^2, b, -b^3]$ which, according to Lemma \ref{p1}, is \numP-hard (as
we assumed just now, it is not $[1,-1,\pm 1,\mp 1]$). Now we assume $a+b^2 \neq 0$. Applying Theorem \ref{2-3} to $g$, it is \numP-hard (and so is the given signature $[1,a,b,ab]$) unless \begin{enumerate}
    \item $X=1$. $(1+ab)(b+bc) = (a+b^2)^2$ implies $\left(a^2-b \right)\left(b^3-1 \right)=0$. If $a^2-b=0$, the given signature is $[1,a,a^2,a^3]$ and  is degenerate. If $b^3-1=0$, since $b\in\mathbb{Q}$, we have $b=1$ and  $a
    \ne  -b^2 = -1$, and the given signature is $[1,a,1,a]$. We apply $G_{aux}$ in Figure \ref{fa} using $g$ to produce a ternary signature $h=[2 + 2a^3, 2a + 2a^2, 2a + 2a^2, 2 + 2a^3]$ on LHS, which has the form $[1,a',a',1]$ after normalization, as $2 + 2 a^3 \ne 0$ for $a \in \mathbb{Q}$ and $a \ne -1$. So, $h$ is \numP-hard by Lemma \ref{easy} unless $a' = 0, \pm 1$, which implies $a=1$ (as $a \ne 0, -1$) in which case the given signature is $[1,1,1,1]$ and thus in P. Thus, $[1,a,1,a]$ ($a\neq 0$) is  \numP-hard unless $a=\pm 1$.
    \item $X=Z=0$. The given signature is $[1,a,-\frac{1}{a}, -1]$ which, by Lemma \ref{p2}, is \numP-hard unless $a=\pm1$.
    \item $X=-1$, $Z=0$. This turns out to be impossible.
    \item $X=-1$, $Z=-1$. This is also impossible.
    \end{enumerate}
Note that since $a,b\ne 0$, $[1,a,b,ab]$ cannot be Gen-Eq. The lemma is  proved.
\end{proof}


\begin{lemma} \label{1a-a-1}
The problem $[1, a, -a, -1]$ with $a\in \mathbb{Q}$ is \numP-hard unless $a=\pm 1, 0$, which is in P.
\end{lemma}
\begin{proof}
If $a=\pm1, 0$, it is easy to see the given signature is in P, with
$[1,1,-1,-1]$ being affine. Now we assume $a\ne \pm 1, 0$. We use the gadget $G_4$ to produce a ternary signature
$g = (1+a) [u, v, v, u]$,
where $u = 1-a+a^2$, and $v = a(1-a^2)$.
Since $u, v \ne 0$ and $u \ne \pm v$,
by Lemma \ref{easy}, $g$ is \numP-hard and so is the given signature $[1,a,-a,-1]$.
\end{proof}

Now we prove
%
\begin{theorem} \label{large5}
The problem $[1,a,b,c]$ with $a,b,c \in \mathbb{Q}$, $abc\ne0$, is \numP-hard unless it is degenerate, Gen-Eq or affine.
\end{theorem}
\begin{proof}
 By Theorem \ref{g1works1} and Lemma \ref{c=ab} it suffices to consider the case when the ratio of two eigenvalues in $G_1=\smm{1}{b}{a}{c}$ is a root of unity and $c \neq ab$. If the ratio of eigenvalues of $G_1$ is a root of unity, we know at least one condition in (\ref{c1eq}) holds. For  convenience, we list the conditions in (\ref{c1eq}) here and label them as $R_i$ where $i=1,2,3,4,5$:
\begin{equation} \label{j1}
    R = \bigvee_{i=1}^{5} R_i,
    \ \ \text{where} \begin{cases}
     R_1: c=-1\\
     R_2: ab + c^2+c+1=0\\
     R_3: 2ab+c^2+1=0\\
     R_4: 3ab+c^2-c+1=0\\
     R_5: 4ab+c^2-2c+1=0\\
    \end{cases}
\end{equation}


We apply $G_{aux}$ in Figure \ref{fa} on $[1,a,b,c]$, i.e.\ placing squares to be $[1,a,b,c]$ and circles to be $=_3$, to produce a ternary signature $[w,x,y,z]=[1 + 2a^3 + b^3,a + 2a^2b + b^2c,a^2 + 2ab^2 + bc^2, a^3 + 2b^3 + c^3]$. If $w \neq 0$ and $G_1$ works on $[w,x,y,z]$, by Theorem \ref{g1works1} we have $[w,x,y,z]$ is \numP-hard and thus $[1,a,b,c]$ is \numP-hard unless at least one condition $S_i$ listed below holds, where $i=1,2,3,4,5,6$:
\begin{equation} \label{j2}
    S = \bigvee_{i=1}^{6} S_i
    \ \ \text{where} \begin{cases}
    S_1 : x^2=wy \wedge y^2=xz\ (\text{degenerate form})\\
    S_2 : x=0 \wedge y=0\ (\text{Gen-Eq form})\\
    S_3 : w=y \wedge x=0 \wedge z=0\ (\text{affine form $[1,0,1,0]$}) \\
    S_4 : w+y=0 \wedge x = 0 \wedge z=0\ (\text{affine form $[1,0,-1,0]$})\\
    S_5 : w=x \wedge w + y = 0 \wedge w+z =0\ (\text{affine form $[1,1,-1,-1]$}) \\
    S_6 : w+x=0 \wedge w+y=0 \wedge w=z\ (\text{affine form $[1,-1,-1,1]$})\\
    \end{cases}
\end{equation} 
\noindent Solve the equation system $R \wedge S$ for variables $a,b,c \in \mathbb{Q}$, we have the following solutions:
\begin{itemize}
    \item $a=c=-1, b=1$; the problem $[1,-1,1,-1]$ is in P since it is degenerate;
    \item $a=1, b=c=-1$; the  problem $[1,1,-1,-1]$ is in P since it is affine;
    \item $a=c=1, b=-1$; the problem $[1,1,-1,1]$ is \numP-hard (use the gadget $G_4$ to produce $[1,1,-1,3]$ after flipping 0's and 1's, then use it again to produce $[1,1,-5,19]$ which is \numP-hard by Theorem~\ref{g1works1}. Note that we need to apply
    $G_4$ twice in order that the condition that
    $G_1$ works
    in 
    Theorem~\ref{g1works1} is satisfied for the newly created ternary signature);
    \item $a=\frac{1}{2}, b = -\frac{1}{2}, c=-1$; the problem $[1,\frac{1}{2},-\frac{1}{2},-1]$ is \numP-hard by Lemma \ref{1a-a-1}.
\end{itemize}

Continuing the discussion for the ternary signature $[w,x,y,z]$, it remains to consider the case when $w=0$ or $G_1$ does not work on $[w,x,y,z]$. For $ w \ne 0$ we  normalize $[w,x,y,z]$ to be $[1,\frac{x}{w}, \frac{y}{w}, \frac{z}{w}]$ and substituting $\frac{x}{w}, \frac{y}{w}, \frac{z}{w}$ into $a,b,c$ respectively in (\ref{c1eq}), we get at least one condition $T_i$ listed below, where $i=1,2,3,4,5,6$:
\begin{equation} \label{j3}
    T = \bigvee_{i=1}^{6} T_i,
    \ \ \text{where} \begin{cases}
    T_1 : zw+w^2=0\\
    T_2 : xy+z^2+zw+ w^2=0\\
    T_3 : 2xy+z^2+w^2=0\\
    T_4 : 3xy+z^2-zw+w^2=0\\
    T_5 : 4xy+z^2-2zw+w^2=0\\
    T_6 : xy=wz\\
    \end{cases}
\end{equation}

Note that $T_1$ incorporates the case when $w=0$. So we have the condition $R \wedge T$.
We now apply  $G_{aux}$ once
again using $[w,x,y,z]$ to produce another new ternary signature $[w_2,x_2,y_2,z_2]$ where $w_2 = w^3+2x^3+y^3$, $x_2=w^2x+2x^2y+y^2z$, $y_2=wx^2+2xy^2+yz^2$, $z_2=x^3+2y^3+z^3$. Similarly as the previous argument, if $w_2 \neq 0$ and $G_1$ works on $[w_2,x_2,y_2,z_2]$, we know $[w_2,x_2,y_2,z_2]$ is \numP-hard and thus $[1,a,b,c]$ is \numP-hard unless at least one condition $U_i$  listed below holds, where $i=1,2,3,4,5,6$:
\begin{equation} \label{j4}
    U = \bigvee_{i=1}^{6} U_i,
    \ \ \text{where} \begin{cases}
    U_1 : x_2^2=w_2y_2 \wedge y_2^2=x_2z_2\ (\text{degenerate form})\\
    U_2 : x_2=0 \wedge y_2=0\ (\text{Gen-Eq form})\\
    U_3 : w_2=y_2 \wedge x_2=0 \wedge z_2=0\ (\text{affine form $[1,0,1,0]$}) \\
    U_4 : w_2+y_2=0 \wedge x_2= 0 \wedge z_2=0\ (\text{affine form $[1,0,-1,0]$})\\
    U_5 : w_2=x_2 \wedge w_2 + y_2 = 0 \wedge w_2+z_2 =0\ (\text{affine form $[1,1,-1,-1]$}) \\
    U_6 : w_2+x_2=0 \wedge w_2+y_2=0 \wedge w_2=z_2\ (\text{affine form $[1,-1,-1,1]$})\\ \end{cases}
\end{equation}

Solve the equation system $R \wedge T \wedge U$ for rational-valued variables $a,b,c$, we have the following solutions:
\begin{itemize}
    \item $a=c=-1, b=1$; the problem $[1,-1,1,-1]$ is in P since it is degenerate;
    \item $a=1, b=c=-1$; the problem $[1,1,-1,-1]$ is in P since it is affine;
    \item $a=-1, b=c=1$; the problem $[1,-1,1,1]$ is \numP-hard (use the gadget $G_4$ to produce $[1,1,-1,3]$, use it again to produce $[1,1,-5,19]$ which is \numP-hard by Theorem~\ref{g1works1});
    \item $a=c=1, b=-1$; the problem $[1,1,-1,1]$ is \numP-hard (this is the reversal of $[1,-1, 1, 1]$);
    \item $a=\frac{1}{2}, b = -\frac{1}{2}, c=-1$; the problem $[1,\frac{1}{2},-\frac{1}{2},-1]$ is \numP-hard by Lemma \ref{1a-a-1}.
\end{itemize}

Otherwise, we know $w_2 = 0$ or $G_1$ does not work on $[w_2,x_2,y_2,z_2]$. Similarly, we know at least one condition $V_i$  listed below holds, where $i=1,2,3,4,5, 6$:
\begin{equation}\label{j5}
    V =    \bigvee_{i=1}^{6} V_i,
    \ \ \text{where} \begin{cases}
    V_1 : z_2w_2+w_2^2=0\\
    V_2 : x_2y_2+z_2^2+z_2w_2+ w_2^2=0\\
    V_3 : 2x_2y_2+z_2^2+w_2^2=0\\
    V_4 : 3x_2y_2+z_2^2-z_2w_2+w_2^2=0\\
    V_5 : 4x_2y_2+z_2^2-2z_2w_2+w_2^2=0\\
    V_6 : x_2y_2=w_2z_2\\
    \end{cases}
\end{equation}

Finally, solve the equation system $R\wedge T\wedge V$ for variables $a,b,c \in \mathbb{Q}$, we have the following solutions:
\begin{itemize}
    \item $a=-1, b=c=1$; the problem $[1,-1,1,1]$ is \numP-hard (see the case above for $R \wedge T \wedge U$);
    \item $a=c=1, b=-1$; the problem $[1,1,-1,1]$ is \numP-hard (this is the reversal of $[1,-1,1,1]$);
    \item $a=-b, c=-1$; the problem $[1,a,-a,-1]$ is \numP-hard unless $a=\pm1, 0$ by Lemma \ref{1a-a-1}.
\end{itemize}

The proof is now complete.

\end{proof}

\section{Dichotomy for $[0,a,b,0]$} 
We quickly finish the discussion for $[0,a,b,0]$ with the help of previous theorems on $[1,a,b,c]$.
\begin{theorem} \label{0ab0}
The problem $[0,a,b,0]$ with $a,b \in \mathbb{Q}$ is \numP-hard unless $a=b=0$, in which case the Holant value is 0.
\end{theorem}
\begin{proof}
 We apply the gadget $G_4$ on $[0,a,b,0]$ to produce the ternary signature $g=[3a^2b^2, a(a^3+2b^3), b(2a^3+b^3), 3a^2b^2]$. 
 
 If $ab\neq0$, we can normalize $g$  to be the form $[1, a', b', c']$. Since $a, b, c \in \mathbb{Q}$, we have
 $a'b'c' \ne 0$. By Theorem \ref{large5}, we know $[1,a',b',c']$ is \numP-hard (and so is $[0,a,b,0]$) unless it is degenerate, Gen-Eq or affine. However, that $[1,a',b',c']$ in P implies $b=a$, i.e.\ the given signature is $[0,a,a,0]$. It suffices to consider $[0,1,1,0]$. We apply the gadget $G_{aux}$ on $[0,1,1,0]$ and get the ternary signature $[3,2,2,3]$ which is \numP-hard by Lemma \ref{easy}, so $[0,1,1,0]$ is \numP-hard.
 
 Now if exactly one of $a$ and $b$ is 0, it suffices to consider the problem $[0,1,0,0]$. This problem is to count the number of exact set covers in a 3-regular 3-uniform set system.  This problem is \numP-hard by Lemma 6.1 in~\cite{fan2020dichotomy}. 
\end{proof}

\section{Main Theorem}
We are now ready to prove our main theorem. The following is a flowchart of the logical structure for the proof of
Theorem~\ref{thm:main}.

\bigskip  

\textbf{Flowchart of proof structure: }

\medskip

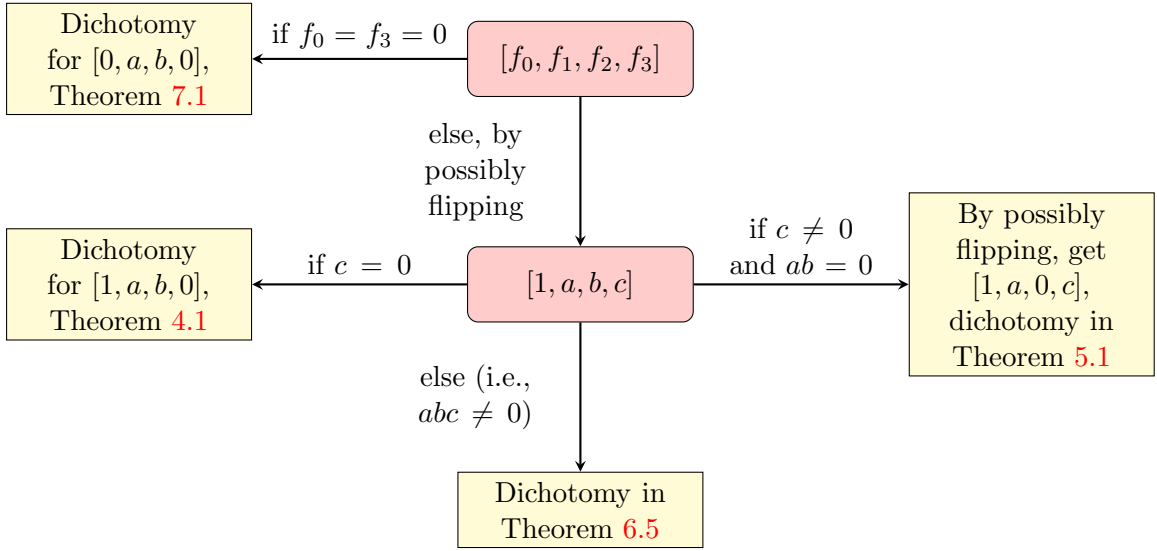
\tikzstyle{startstop} = [rectangle, rounded corners, minimum width=3cm, minimum height=1cm,text centered, draw=black, fill=red!20]
\tikzstyle{io} = [trapezium, trapezium left angle=70, trapezium right angle=110, minimum width=3cm, minimum height=1cm, text centered, draw=black, fill=blue!20]
\tikzstyle{process} = [rectangle, minimum width=3cm, minimum height=1cm, text centered,text width=3cm, draw=black, fill=yellow!20]
\tikzstyle{decision} = [diamond, minimum width=3cm, minimum height=1cm, text centered, draw=black, fill=green!20]
\tikzstyle{arrow} = [thick,->,>=stealth,  text centered, text width=2.5cm]

\begin{tikzpicture}[node distance=3cm]
    \node (s1) [startstop] {$[f_0,f_1,f_2,f_3]$};
    \node (s2) [process, left of=s1, xshift=-3cm] {Dichotomy for $[0,a,b,0]$, Theorem~\ref{0ab0}};
    \node (s3) [startstop, below of=s1] {$[1,a,b,c]$};
    \node (s4) [process, left of=s3,xshift=-3cm] {Dichotomy for $[1,a,b,0]$, Theorem~\ref{1ab0}};
    \node (s5) [process, right of=s3, xshift=3cm] {By possibly flipping, get $[1,a,0,c]$, dichotomy in Theorem~\ref{1a0c}};
    \node (s6) [process, below of=s3] {Dichotomy in Theorem~\ref{large5}};
    
    \draw [arrow] (s1) -- node[anchor=south] {if $f_0=f_3=0$} (s2);
    \draw [arrow] (s1) -- node[anchor=east] {else, by possibly flipping} (s3);
    \draw [arrow] (s3) -- node[anchor=south] {if $c=0$} (s4);
    \draw [arrow] (s3) -- node[anchor=south] {if $c \ne 0$ and $ab=0$} (s5);
    \draw [arrow] (s3) -- node[anchor=east] {else (i.e., $abc\ne 0$)} (s6);
\end{tikzpicture}

\begin{theorem}\label{thm:main}
The problem $\Holant\{\,[f_0,f_1,f_2,f_3]\, |\, (=_3)\}$ with $f_i \in \mathbb{Q}$ $(i = 0,1,2,3)$ is \numP-hard unless 
the signature $[f_0, f_1,f_2,f_3]$ is degenerate, Gen-Eq or belongs to the affine class.
\end{theorem}
 \begin{proof}
 First, if $f_0 = f_3 = 0$, by Theorem \ref{0ab0}, we know that it is  \numP-hard unless it is $[0,0,0,0]$ which is degenerate. Note that in all other cases, $[0,f_1,f_2,0]$ is not Gen-Eq, degenerate or  affine.
 
 Assume now at least one of $f_0$ and $f_3$ is not $0$. By flipping the role of 0 and 1, we can assume $f_0\ne 0$, then the signature becomes $[1,a,b,c]$ after normalization. 
 
 If $c=0$, the dichotomy for $[1,a,b,0]$ is proved in Theorem \ref{1ab0}.
 
 If in $[1,a,b,c]$, $c\ne 0$, then $a$ and $b$ are symmetric by flipping. Now if $ab=0$, we can assume $b=0$ by the afore-mentioned symmetry, i.e., the signature becomes $[1,a,0,c]$. By Theorem \ref{1a0c}, it is \numP-hard unless $a=0$, in which case it is Gen-Eq. In all other cases, it is not Gen-Eq or degenerate or  affine.
 
 Finally, for the problem $[1,a,b,c]$ where $abc\ne0$, Theorem \ref{large5} proves the dichotomy that it is \numP-hard unless the signature is degenerate or Gen-Eq or  affine.
 \end{proof}

\newpage
\bibliographystyle{plain}
\bibliography{References}

\newpage
\section*{Appendix}
In this paper we use Mathematica\texttrademark{}
to carry out symbolic computations.
 In particular, the function \texttt{CylindricalDecomposition} is used heavily. It is an implementation  of Tarski’s theorem  on  the  decidability  of  the  theory  of  real-closed  fields, and  
 can prove the non-existence solutions of polynomial systems. 

\vspace{.1in}

\noindent
$\bullet$ In the proof of Lemma~\ref{3.1.2}, we use
\texttt{CylindricalDecomposition}
to show that the intersection 
 of $f_1=f_2=f_3=f_4=0$  is empty for $a, b, c \in \mathbb{Q}$ and $x \in \mathbb{C}$,  where $f_1 = cx^3+3bx^2+3ax+1$, $f_2=(ab+c)x^3+ (3bc+2a^2+b)x^2 +(2b^2+ac + 3a)x+ab+1$, $f_3=(a^3+b^3+c^3)x^3+3(a^2+2ab^2+bc^2)x^2 + 3(a+2a^2b+b^2c)x + 1 + 2a^3+b^3$ and $f_4=(ab+2abc+c^3)x^3+(2a^2+b+2a^2c+3ab^2+bc+3b^2c)x^2+(3a+3a^2b+ac+2b^2+2b^2c+ac^2)x+1+2ab+abc$. We write $x$ as $u+vi$ with $u, v \in \mathbb{R}$, and set both the real and imaginary parts of each $f_i$ to 0. 

\begin{lstlisting}[language=Mathematica]
Clear["Global`*"];
(* below x = u + vi is complex, 
r3 is real part of x^3, i3 is imaginary part of x^3 , 
r2 is real part of x^2, i2 is imaginary part of x^2, 
r1 is real part of x, i1 is imaginary part of x *)

r3 = u^3 - 3 u v v;   
i3 = 3 u u v - v^3;
r2 = u u - v v;
i2 = 2 u v;
r1 = u;
i1 = v;
(* below f_k1 is real part of f_k, f_k2 is imaginary part of f_k, 
where k = 1,2,3,4 *)
f11 = c r3 + 3 b r2 + 3 a r1 + 1;
f12 = c i3 + 3 b i2 + 3 a i1;
f21 = (a b + c) r3 + (3 b c + 2 a a + b) r2 + (2 b b + a c + 3 a) r1 +
    a b + 1;
f22 = (a b + c) i3 + (3 b c + 2 a a + b) i2 + (2 b b + a c + 3 a) i1;
f31 = (a^3 + b^3 + c^3) r3 + 3 (a a + 2 a b b + b c c) r2 + 
   3 (a + 2 a a b + b b c) r1 + 2 a^3 + b^3 + 1;
f32 = (a^3 + b^3 + c^3) i3 + 3 (a a + 2 a b b + b c c) i2 + 
   3 (a + 2 a a b + b b c) i1;
f41 = (a b + 2 a b c + c^3) r3 + (2 a a + b + 3 a b b + 2 a a c + 
      b c + 3 b c c) r2 + (3 a + 3 a a b + 2 b b + a c + 2 b b c + 
      a c c) r1 + 1 + 2 a b + a b c;
f42 = (a b + 2 a b c + c^3) i3 + (2 a a + b + 3 a b b + 2 a a c + 
      b c + 3 b c c) i2 + (3 a + 3 a a b + 2 b b + a c + 2 b b c + 
      a c c) i1;

CylindricalDecomposition[
 f11 == 0 && f12 == 0 && f21 == 0 && f22 == 0 && f31 == 0 && 
  f32 == 0 && f41 == 0 && f42 == 0, {a, b, c, u, v}]

\end{lstlisting}

\vspace{.1in}

\noindent
$\bullet$ In the proof of Lemma~\ref{1-1 X=1}, we use \texttt{CylindricalDecomposition} several times to solve a polynomial system in $a,b,c$. There,  \emph{(con1)} is $a^3 - b^3 - ab(1-c) =0$ and
\emph{(con2)} is $a^3+ab+2b^3=0$. Together with a third condition, we solve for $a,b,c \in \mathbb{Q}$. The third equation is among one in (\ref{g3works}) (note that we use the function \texttt{Factor} here to get an irreducible polynomial $f_1$ over $\mathbb{Q}$), or one in (\ref{c3eq}), or $(a+b^2)(a^2+bc) = (1+ab)(ab+c^2) $.  

\begin{lstlisting}[language=Mathematica]
Clear["Global`*"];
(* dsq means\Delta^2 *)

dsq = 1 + 4 a^3 + 4 a^2 b^2 + 4 a b c + 4 b^3 c - 2 c^2 + c^4;
d = 1 - c c;
e = 2 (a + b b);

(*y^3 + y^2 b + y a + c\[Equal]0 is transformed to LHS = RHS below,
by eliminating the square root part 
i.e., we give a function in real domain *)

LHS = dsq*(dsq + 3 d d + 2 d e b + e e a)^2;
RHS = (dsq*(3 d + e b) + d d d + d d e b + e e a d + c e e e)^2;
Factor[LHS - RHS]
(* the factor result is f1 * (a + b b) ^3 , where f1 is below    *)

f1 = a^3 + 4 a^6 + 3 a^5 b^2 + a^3 b^3 - c - 4 a^3 c + 6 a^4 b c - 
   6 a^2 b^2 c - b^3 c - 3 a^2 b^5 c - 3 a^3 c^2 - 3 a b c^2 - 
   4 b^3 c^2 - a^3 b^3 c^2 - 6 a b^4 c^2 - 4 b^6 c^2 + 3 c^3 + 
   4 a^3 c^3 + 6 a^2 b^2 c^3 + 3 b^3 c^3 + a^3 c^4 + 3 a b c^4 + 
   4 b^3 c^4 - 3 c^5 - b^3 c^5 + c^7;
f2 = a + b b;
f3 = c - a b;
f4 = a^3 - b^3 c + a b (c c - 1);
con1 = a a a - b b b - a b (1 - c) == 0;
con2 = a a a + a b + 2 b b b == 0;
(* the following four commands correspond to (3.7) in paper,
the exceptional cases of that G3 works but may not interpolate any \
unary, 
 *)
CylindricalDecomposition[con1 && con2 && f1 == 0, {a, b, c}]
CylindricalDecomposition[con1 && con2 && f2 == 0, {a, b, c}]
CylindricalDecomposition[con1 && con2 && f3 == 0, {a, b, c}]
CylindricalDecomposition[con1 && con2 && f4 == 0, {a, b, c}]

A = 1 + 2 a b + c^2;
B = 4 a^3 + 4 a^2 b^2 + 4 a b c + 4 b^3 c + (-1 + c^2)^2;
h1 = A == 0;
h2 = B == 0;
h3 = A A + B == 0;
h4 = A A + 3 B == 0;
h5 = 3 A A + B == 0;
(* G3 doesn't work *)
CylindricalDecomposition[con1 && con2 && h1, {a, b, c}]
CylindricalDecomposition[con1 && con2 && h2, {a, b, c}]
CylindricalDecomposition[con1 && con2 && h3, {a, b, c}]
CylindricalDecomposition[con1 && con2 && h4, {a, b, c}]
CylindricalDecomposition[con1 && con2 && h5, {a, b, c}]
CylindricalDecomposition[
 con1 && con2 && a^3 - a b + a b c + b^3 c - c^2 - a b c^2 == 0, {a, 
  b, c}]
\end{lstlisting}

\vspace{.1in}

\noindent
$\bullet$ In the proof of Lemma~\ref{1-2 X=1}, similarly, we have the code below.
\begin{lstlisting}[language=Mathematica]
Clear["Global`*"];
(* note that the following con2 is not the same one as in the lemma, 
however, according to the proof of the lemma,
it suffices to consider this new con2 with the original con1 *)

con1 = a a a - b b b - a b (1 - c) == 0;
con2 = a^9 + a^4 b^4 + a^3 b^6 + b^9 == 0;

f1 = a^3 + 4 a^6 + 3 a^5 b^2 + a^3 b^3 - c - 4 a^3 c + 6 a^4 b c - 
   6 a^2 b^2 c - b^3 c - 3 a^2 b^5 c - 3 a^3 c^2 - 3 a b c^2 - 
   4 b^3 c^2 - a^3 b^3 c^2 - 6 a b^4 c^2 - 4 b^6 c^2 + 3 c^3 + 
   4 a^3 c^3 + 6 a^2 b^2 c^3 + 3 b^3 c^3 + a^3 c^4 + 3 a b c^4 + 
   4 b^3 c^4 - 3 c^5 - b^3 c^5 + c^7;
f2 = a + b b;
f3 = c - a b;
f4 = a^3 - b^3 c + a b (c c - 1);

CylindricalDecomposition[con1 && con2 && f1 == 0, {a, b, c}]
CylindricalDecomposition[con1 && con2 && f2 == 0, {a, b, c}]
CylindricalDecomposition[con1 && con2 && f3 == 0, {a, b, c}]
CylindricalDecomposition[con1 && con2 && f4 == 0, {a, b, c}]

A = 1 + 2 a b + c^2;
B = 4 a^3 + 4 a^2 b^2 + 4 a b c + 4 b^3 c + (-1 + c^2)^2;

h1 = A == 0;
h2 = B == 0;
h3 = A A + B == 0;
h4 = A A + 3 B == 0;
h5 = 3 A A + B == 0;

CylindricalDecomposition[con1 && con2 && h1, {a, b, c}]
CylindricalDecomposition[con1 && con2 && h2, {a, b, c}]
CylindricalDecomposition[con1 && con2 && h3, {a, b, c}]
CylindricalDecomposition[con1 && con2 && h4, {a, b, c}]
CylindricalDecomposition[con1 && con2 && h5, {a, b, c}]
CylindricalDecomposition[
 con1 && con2 && a^3 - a b + a b c + b^3 c - c^2 - a b c^2 == 0, {a, 
  b, c}]
\end{lstlisting}

\vspace{.1in}
\noindent
$\bullet$ In the proof of Theorem~\ref{large5}, again we use \texttt{CylindricalDecomposition} to solve equation systems $R \wedge S$, $R \wedge T \wedge U$ and $R \wedge T \wedge V$. In solving the system $R \wedge T \wedge V$ which requires a significant amount of computation, 
we apply \texttt{CylindricalDecomposition} 
to each sub-system $R_i \wedge T_j \wedge V_k$ where $i \in \{1,2,3,4,5\}, j,k \in \{1,2,3,4,5,6\}$ separately by using the function \texttt{Manipulate} and combine their solutions. There are a total of $5 \times 6 \times 6 = 180$ sub-systems in $R \wedge T \wedge V$. 

\begin{lstlisting}[language=Mathematica]
Clear["Global`*"];
w = 1 + 2 a a a + b b b;
x = a + 2 a a b + b b c;
y = a a + 2 a b b + b c c;
z = a a a + 2 b b b + c c c;

w2 = w^3 + 2 x^3 + y^3;
x2 = w w x + 2 x x y + y y z;
y2 = w x x + 2 x y y + y z z;
z2 = x x x + 2 y y y + z z z;

tricon = a b c != 0;
(* R: G1 doesn't work on [1,a,b,c] *)
R1 = c == -1;
R2 = a b + c c + c + 1 == 0;
R3 = 2 a b + c c + 1 == 0;
R4 = 3 a b + c c - c + 1 == 0;
R5 = 4 a b + c c - 2 c + 1 == 0;
R = R1 || R2 || R3 || R4 || R5;

(* T: G1 doesn't work on [w,x,y,z] *)
T1 = z w + w w == 0;
T2 = x y + z z + z w + w w == 0;
T3 = 2 x y + z z + w w == 0;
T4 = 3 x y + z z - z w + w w == 0;
T5 = 4 x y + z z - 2 z w + w w == 0;
T6 = x y == w z;
T = T1 || T2 || T3 || T4 || T5 || T6;


(* S:[w,x,y,z] in P *)
S1 = x x == w y && y y == x z;
S2 = x == 0 && y == 0;
S3 = w == y && x == 0 && z == 0;
S4 = w + y == 0 && x == 0 && z == 0;
S5 = w == x && w + y == 0 && w + z == 0;
S6 = w + x == 0 && w + y == 0 && w == z;
S = S1 || S2 || S3 || S4 || S5 || S6;

(* U:[w2,x2,y2,z2] in P *)
U1 = x2 x2 == w2 y2 && y2 y2 == x2 z2;
U2 = x2 == 0 && y2 == 0;
U3 = w2 == y2 && x2 == 0 && z2 == 0;
U4 = w2 + y2 == 0 && x2 == 0 && z2 == 0;
U5 = w2 == x2 && w2 + y2 == 0 && w2 + z2 == 0;
U6 = w2 + x2 == 0 && w2 + y2 == 0 && w2 == z2;
U = U1 || U2 || U3 || U4 || U5 || U6;

CylindricalDecomposition[tricon && R && S, {a, b, c}]

CylindricalDecomposition[tricon && R && T && U, {a, b, c}]


(* Below, con1, con2, con3 corresponds to R, T, V in paper, \
respectively *)
Manipulate[
 CylindricalDecomposition[
  con1 && con2 && con3 && (a b c != 0), {a, b, c}], {con1, {c == -1, 
   a b + c c + c + 1 == 0, 2 a b + c c + 1 == 0, 
   3 a b + c c - c + 1 == 0, 
   4 a b + c c - 2 c + 1 == 0}}, {con2, {z w + w w == 0, 
   x y + z z + z w + w w == 0, 2 x y + z z + w w == 0, 
   3 x y + z z - z w + w w == 0, 4 x y + z z - 2 z w + w w == 0, 
   x y == w z}}, {con3, {z2 w2 + w2 w2 == 0, 
   x2 y2 + z2 z2 + z2 w2 + w2 w2 == 0, 2 x2 y2 + z2 z2 + w2 w2 == 0, 
   3 x2 y2 + z2 z2 - z2 w2 + w2 w2 == 0, 
   4 x2 y2 + z2 z2 - 2 z2 w2 + w2 w2 == 0, x2 y2 == w2 z2}}]
   
\end{lstlisting}


\end{document}